\documentclass[pdflatex,sn-basic]{sn-jnl}

\usepackage{float}
\usepackage{graphicx}%
\usepackage{multirow}%
\usepackage{amsmath,amssymb,amsfonts}%
\usepackage{amsthm}%
\usepackage[title]{appendix}%
\usepackage{xcolor}%
\usepackage{siunitx} 
\usepackage{textcomp}%
\usepackage{manyfoot}%
\usepackage{booktabs}%
\usepackage{algorithm}%
\usepackage{algorithmicx}%
\usepackage{algpseudocode}%
\usepackage{listings}%
\usepackage{comment}
\usepackage[utf8]{inputenc}
\usepackage[T1]{fontenc}
\usepackage{newtxtext,newtxmath}
\usepackage[scr=rsfso]{mathalfa}
\DeclareMathAlphabet{\mathscr}{U}{rsfso}{m}{n}%
\usepackage{tabularray}
\usepackage{caption}
\usepackage{academicons}
\usepackage{hyperref}
\usepackage{textpos}
\usepackage{tabularray}
\usepackage[justification=centering, singlelinecheck=false]{caption}
\usepackage{textcomp}   

\usepackage{amsmath}
\usepackage{tikz}
\usepackage{mathdots}
\usepackage{cancel}
\usepackage{color}
\usepackage{siunitx}
\usepackage{array}
\usepackage{multirow}
\usepackage{amssymb}
\usepackage{gensymb}
\usepackage{tabularx}
\usepackage{extarrows}
\usepackage{booktabs}
\usetikzlibrary{fadings}
\usetikzlibrary{patterns}
\usetikzlibrary{shadows.blur}
\usetikzlibrary{shapes}
\usepackage{xcolor}
\usepackage{amsmath}  
\newcommand{\abs}[1]{\left| #1 \right|}

\usepackage{tabularray} 

\theoremstyle{thmstyleone}%
\newtheorem{theorem}{Theorem}
\theoremstyle{thmstyletwo}%

\theoremstyle{thmstylethree}%

\newtheorem{Lemma}{Lemma}

\begin{document}

\title[Modeling the Prey-Predator Dynamics Leading to Ecological Disaster]
{Modeling the Prey-Predator Dynamics of Habu Snakes and Mongooses Leading to Ecological Disaster on Amami Oshima Island in Japan}

\author[1,2]{\fnm{Pulak} \sur{Kundu}}

\author*[1]{\fnm{Uzzwal Kumar} \sur{Mallick}}\email{mallickuzzwal@math.ku.ac.bd}

\affil[1]{\orgdiv{Mathematics Discipline}, \orgname{Khulna University}, \orgaddress{\street{Khulna-9208}, \country{Bangladesh}}}
\affil[2]{\orgdiv{Department of Mathematics}, \orgname{Bangladesh Army University of Science and Technology Khulna}, \orgaddress{\street{Khulna-9204}, \country{Bangladesh}}}
\abstract{The introduction of mongooses from  Indian subcontinent to Amami Oshima Island, Japan, aimed at controlling the population of venomous Habu snakes, has led to significant ecological disruptions, raising concerns about the long-term sustainability of the island’s biodiversity. To highlight the unintended consequences of such interventions and the necessity of understanding predator-prey dynamics in preserving ecological balance, a mathematical model incorporating snake, mongooses, mouse and natural resources has been proposed to explore their role in the ongoing ecological disaster and analysis the other scenarios if the authorities applied different approaches in place of already implemented strategy. Determining the model's existence and uniqueness, stability at equilibrium points, and state variable characteristics are some of the parts of the analytical analysis of the model. Additionally, sensitivity analysis is conducted to identify sensitive factors. In addition, the Runge-Kutta 4th order has been used to execute the numerical simulations. Our research reveals that although the government began killing and trapping mongooses almost 20 years after their introduction, but if trapping had started just 10 years after their introduction, the outcome could have been drastically different. This time, mongooses would not have been extinct, and their coexistence with other native species would have helped to preserve the ecological balance and prevent the severe ecological damage that is presently being seen. Thus, It is recommended to use mathematical modeling to explore alternatives before decision-making, ensuring sustainable ecosystem management while preventing irreversible impacts of invasive species.}



\maketitle

\section{Introduction}
Introduced in 1979 to manage poisonous habu snakes and crop-damaging rats, the small Indian mongoose (Herpestes javanicus) became the centre of attention of an ecological disaster on Amami Oshima Island, a treasure trove of natural resources. Rather than providing a solution, the mongooses caused a great deal of damage to native species, which led to a decline in biodiversity and the threat of the environment.

Amami Ōshima Island, part of the Amami archipelago in Kagoshima Prefecture, Japan, is renowned for its unique biodiversity, which includes rare and endemic species like the Amami rabbit (\textit{Pentalagus furnessi}), the Amami woodcock (\textit{Scolopax mira}), and the Ryukyu long-tailed giant rat (\textit{Diplothrix legata}) \citep{WikipediaAmami2025}. The island’s subtropical forests, which became a UNESCO World Heritage site in 2021, provide a critical habitat for these species \citep{unesco2021}. However, this ecological haven faced a growing threat from venomous habu snakes (\textit{Protobothrops flavoviridis}), which posed a significant danger to both humans and livestock. The habu snake, known for its aggressive behavior and potentially lethal bite, was responsible for frequent attacks, with over 100 snakebite incidents reported annually on the island by the 1970s \citep{HabuSnakeBite1987, Wakisaka1978}. Moreover, although an antidote for the snake's venom was available, it was challenging to access in rural areas, resulting in numerous fatalities and compelling many residents to leave the island. Additionally, rats were damaging crops, further exacerbating the concerns of local farmers \citep{Yining}. In an effort to control the habu snake and rat populations, authorities decided in 1979 to introduce small Indian mongooses (\textit{Herpestes javanicus}) to the island\citep{Myers2000,Watari2008}. The small Indian mongoose, native to the Indian subcontinent, was considered suitable because of its adaptability to warm climates and its known efficacy in controlling rats and snakes in other regions, such as the Caribbean and Hawaii \citep{Pimentel1955}. Approximately 30 mongooses were imported from India and released on Amami Ōshima under the assumption that they would prey on the nocturnal habu snakes and reduce their threat to humans and agriculture \citep{Ishii2003}.

\begin{figure}[htbp]
    \centering
    \begin{subfigure}[t]{0.46\textwidth}
        \includegraphics[width=1.1\linewidth]{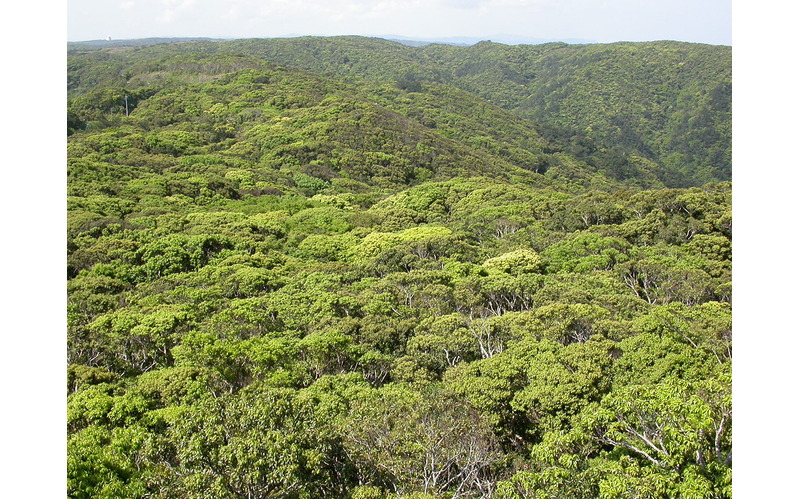}
        \caption{}
        \label{fig:sub11}
    \end{subfigure}
    \hfill
    \begin{subfigure}[t]{0.45\textwidth}
        \includegraphics[width=1.06\linewidth]{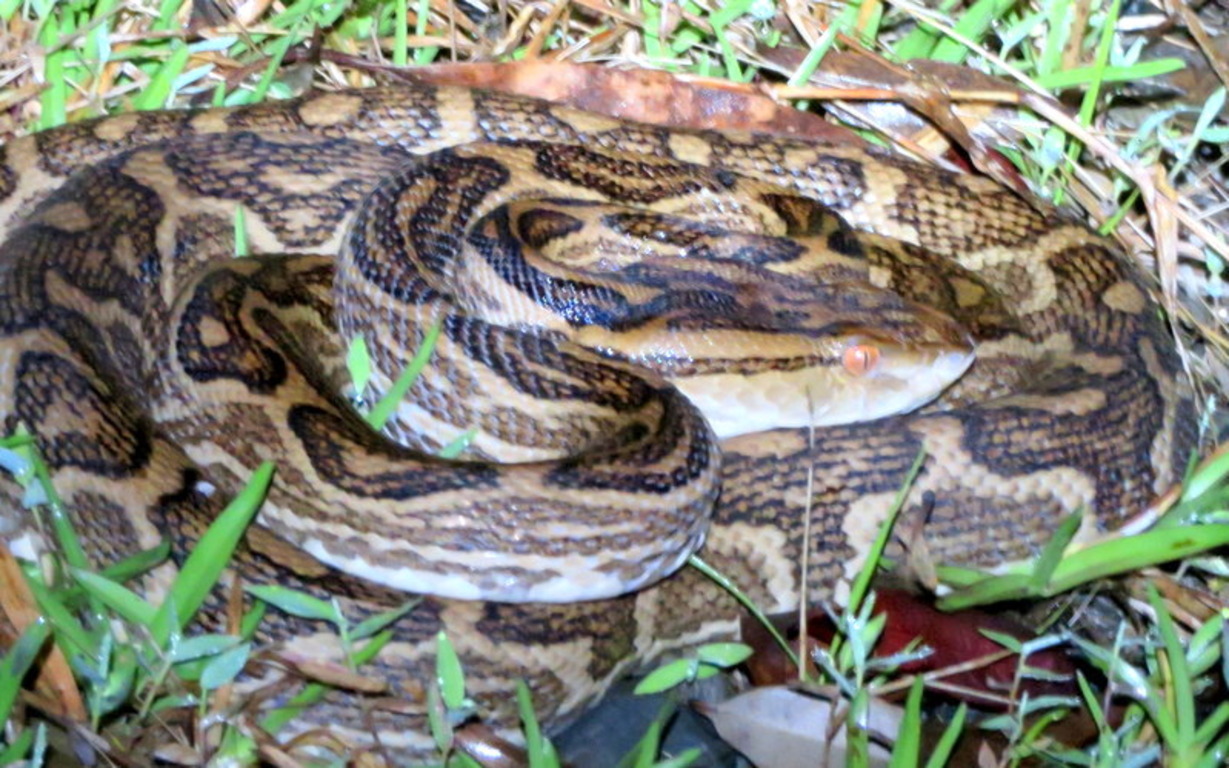}
        \caption{}
        \label{fig:sub22}
    \end{subfigure}
    \vspace{1em} 
    \begin{subfigure}[t]{0.4\textwidth}
        \centering
        \includegraphics[width=\linewidth]{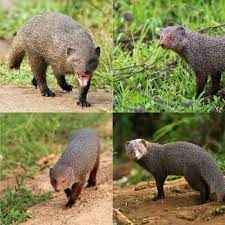}
        \caption{}
        \label{fig:sub33}
    \end{subfigure}
    \caption{(a) Amami Ōshima Island, part of Japan, was famous for its exceptional biodiversity, featuring numerous rare and native species. (b) The Habu snake was a significant predator in ancient ecosystems, posing a threat to early human settlements while playing a crucial role in maintaining ecological balance. (c) Indian mongooses were introduced to Amai Osamia Island with the belief they would reduce Habu snake threats to humans and agriculture.}
    \label{fig:figure 1}
\end{figure}

\begin{figure}[htbp]
    \centering
    \begin{subfigure}[t]{0.47\textwidth}
        \includegraphics[width=1.01\linewidth]{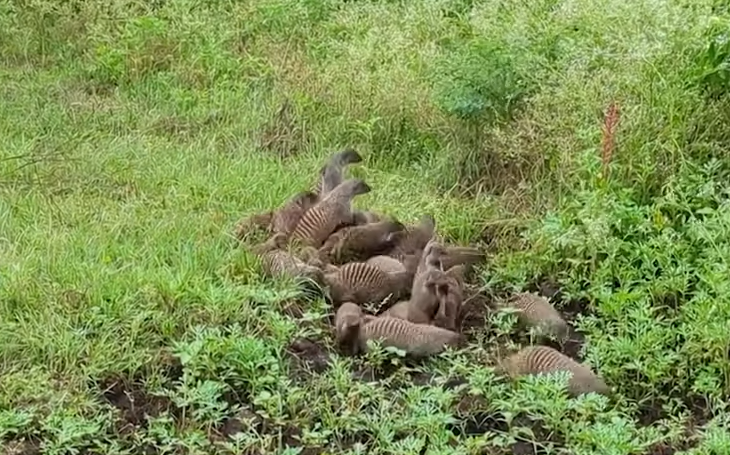}
        \caption{}
        \label{fig:sub1}
    \end{subfigure}
    \hfill
    \begin{subfigure}[t]{0.45\textwidth}
        \includegraphics[width=1.009\linewidth]{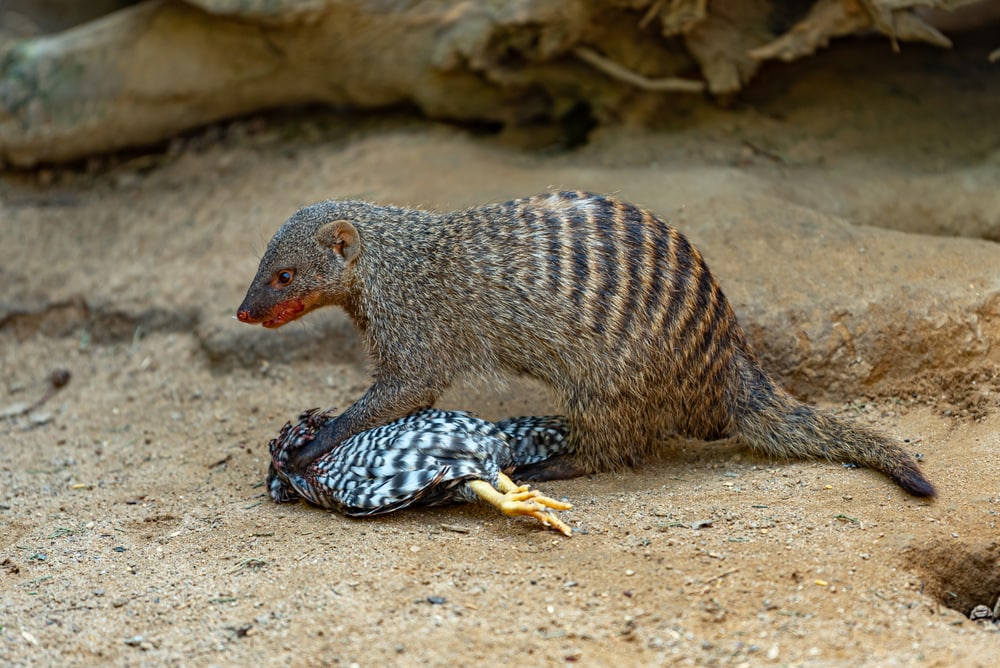}
        \caption{}
        \label{fig:sub2}
    \end{subfigure}
    \vspace{1em} 
    \begin{subfigure}[t]{0.47\textwidth}
        \centering
        \includegraphics[width=1.1\linewidth]{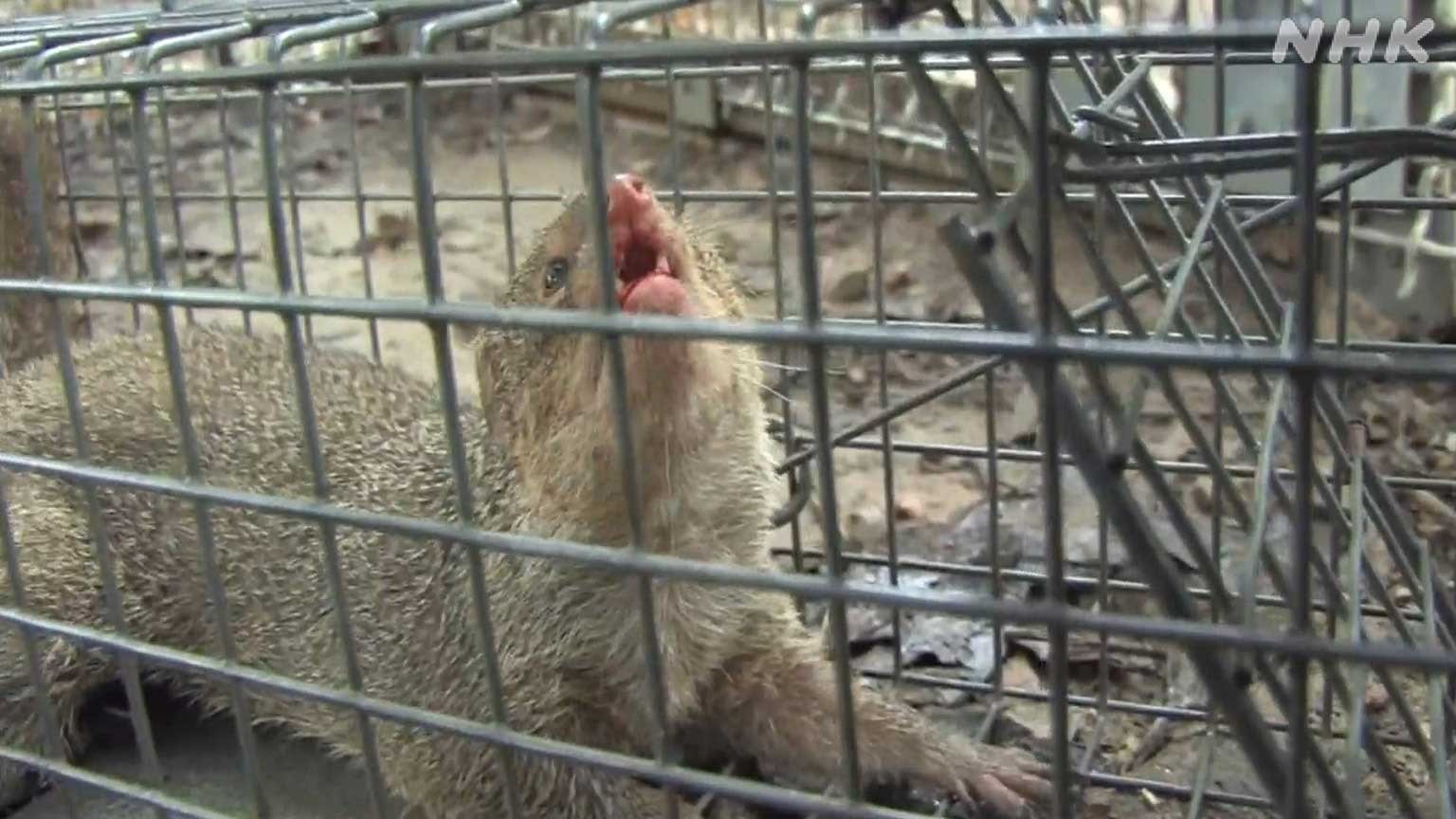}
        \caption{}
        \label{fig:sub3}
    \end{subfigure}
    
    \caption{(a) Mongooses damage crops on Amami Island, showcasing their threat to agriculture.
(b) Mongooses, making harm by preying on the island's natural resources, show as a clear reminder of their disruptive effect on the ecosystem.
(c) Mongooses are caught using a trapping method, demonstrating human efforts to reduce their population.}
    \label{fig:main}
\end{figure}
However, this decision soon backfired due to ecological mismatches. Most mongooses were diurnal (active during the day), while most habu snakes were nocturnal (active at night), resulting in infrequent encounters between the two species \citep{japantimes}. Over time, the mongooses also showed diminishing interest in hunting the snakes after their introduction. Instead of preying on habu snakes, the mongooses turned to hunt native species, many of which had evolved in isolation without natural predators and were ill-equipped to defend themselves \citep{Yamada}. The Amami rabbit, in particular, experienced sharp population declines as mongooses targeted its young. The ecosystem of Amami Ōshima, delicately balanced over millennia, was thrown into disarray, with cascading effects on flora and fauna dependent on the island’s unique biodiversity. Besides, crops in the fields were being damaged by the mongoose, further aggravating the economic hardships faced by local farmers. By the 1990s, the mongoose population had exploded to approximately 10,000 due to the absence of natural predators and favourable environmental conditions causing the ecological imbalance of this island \citep{YamadaSugimura2004}. Thus, this incident highlights that no matter how intelligent humans may be, they can never achieve their own interests by going against nature.

Recognizing the ecological disaster, the Japanese Ministry of the Environment (then the Environment Agency)  and Kagoshima Prefecture launched an eradication program in 2000 to mitigate the damage caused by the mongooses \citep{Schultz2024}. This comprehensive effort involved setting traps, administering poison baits, conducting systematic removal campaigns, and implementing advanced monitoring techniques, such as the use of camera traps, to track mongoose activity and assess the effectiveness of the control measures \citep{Fukasawa2013,Ishida2015}. Over an 18-year period, approximately 32,000 mongooses were removed, with annual captures exceeding 2,000. Since May 2018, no mongooses have been captured, and continuous monitoring has confirmed their absence for more than six years \citep{AsahiShimbun5}. Finally, on September 3, 2024, the Ministry officially declared the eradication of the small Indian mongoose from Amami Ōshima Island \citep{Ministry2024}.

On the other side, ecological dynamics, especially those involving the interactions between native and invading species, can be better understood by mathematical modelling. To study population dynamics across time, a reliable framework is provided by prey-predator models like the Lotka-Volterra equations. Simulating different scenarios, evaluating management measures' impacts, and predicting long-term ecological consequences are all made possible by these models. However, multiple studies by different researchers have been conducted to emphasize the severity of this ecological disaster on the island. \cite{YamadaSugimura2004} emphasize the severe impact of the invasive small Indian mongoose on Amami Ōshima’s native wildlife and highlight the need for increased efforts and resources to achieve successful eradication and ecological restoration. Again,  \cite{Watari2008} investigate the effects of the invasive small Indian mongoose on the native fauna of Amami Ōshima by analyzing distribution patterns along a historical gradient of mongoose invasion. Their research showed that larger native species decreased in areas with mongooses, while smaller species grew in number. This change is due to effects on the food chain, which emphasises the importance of looking at food web structures when judging the ecological risks of invasive predators. Also, \cite{Fukasawa2013} applied a Bayesian state-space modelling approach to estimate the dynamics of invasive mongoose populations on Amami Island, Japan, highlighting the long-term effectiveness of control measures and the feasibility of eradication based on population size and capture efforts. Furthermore, \cite{WanSamperisam2024} presents a comprehensive review on advanced predator-prey models, exploring the incorporation of factors like fractional calculus, role reversal, and spatial heterogeneity, which enhance the understanding of ecological dynamics and contribute to practical solutions in pest control, disaster management, and ecological disorders. A recent predator–prey study by \cite{Kabir2025EcoEpi} demonstrated that incorporating disease dynamics, behavioral strategies, and the Allee effect strongly affects population stability, with low infection rates maintaining equilibrium and high transmission inducing oscillations and extinction risks. Additionally, sandy beach predator–prey system investigations based on Lotka–Volterra modeling and stable isotope evidence indicated that predation can be a primary driver of population dynamics, even in physically harsh environments \citep{ARUEIRA2025107488}.\
In this study, we have formulated a newly proposed mathematical model to capture this phenomena occurring in Amami Ōshima Island, Japan, and have analyzed its dynamics. By exploring alternative strategies to the current approach, we have identified the most effective solution for maintaining ecological balance on this island.\\
The remaining part of the article is structured as follows: In Section \ref{sec:Mathematical Formulation of the Model}, the formulation of a mathematical model is outlined to investigate the prey-predator dynamics of Habu snakes and mongooses, which have led to an ecological disaster on Amami Ōshima Island. In Section \ref{sec:Analytical  Analysis}, a detailed analytical examination of the model is presented, whereas numerical simulations are illustrated in Section \ref{sec:Model Simulations}. The findings, along with key insights, are discussed in Section \ref{sec:Results and disussion}. Finally, the conclusions drawn from this study are provided in Section \ref{sec:Conclusions}.

\section{Mathematical Formulation of the Model }\label{sec:Mathematical Formulation of the Model}
In this section, we have proposed a new mathematical model to investigate the interactions between prey, predators, and natural resources on Amami Oshima Island, Japan, to explore the ecological damage that has occurred and analyze potential outcomes if authorities implemented different strategies in place of the already implemented in that island.

\subsection{Assumptions}
At the time of formulating our newly proposed model, the following assumptions have been considered into account:
\begin{itemize}
\item The populations of mongooses $M(t)$, habu snakes $S(t)$, rats $M(t)$, and natural resources $N_R(t)$ (shown in Table \ref{tab:variable_introdution}) are considered variables in the ecosystem of Amami Oshima Island, Japan, to formulate the model.
 \item The Amami rabbit(\textit{Pentalagus furnessi}), Ryukyu robin(\textit{Larvivora komadori}), Amami woodcock(\textit{Scolopax rusticola}), as well as chickens, particularly their eggs and chicks, which are important parts of the ecological balance of Amami Oshima Island, have been taken as natural resources.
\item The populations of mongooses, snakes, and rats are assumed to follow a Malthusian growth rate, whereas natural resources are considered to grow logistically.
\item In this ecosystem, habu snakes, rats, and natural resources are considered prey for mongooses, while habu snakes prey on rats.
\item The trapping and killing term would be active from where it has been done.
\end{itemize}
\begin{table}[hbt]
    \centering
     \caption{Description of Variable}
    \begin{tabular}{c c}
    \hline 
   \textbf{\textcolor{blue}{Variable}} & \textbf{\textcolor{blue}{Description}} \\
    \hline
         $M(t)$&    Number of Mongooses   \\
   
         $S(t)$ &  Number of Snakes (Habu Snakes) \\
                 $R(t)$& Number of Rats \\
                $N_R(t)$&    Natural Resources (Number of important prey species \\
         &  on the island for ecological balance)  \\
         \hline
    \end{tabular}
    \label{tab:variable_introdution}
\end{table}

\subsection{Formulation}
On Amami Ōshima Island in Japan, the mongoose population follows the Malthusian law, where its growth rate depends on its current size. Originally introduced to control venomous habu snakes, mongooses found an abundance of prey, including rats, native species like birds, amphibians, and small mammals, allowing their numbers to rise. However, their presence soon became a threat to the island’s delicate ecosystem, endangering native wildlife. In response, conservation efforts intensified, with humans employing trapping methods such as baited cages, snares, and poisoning to curb their spread. Considering these factors, the first differential equation governing the mongoose population is given by:
\begin{align}
     \frac{dM}{dt}&=\alpha M+\beta MS+\delta MR+\tau MN_R-(\mu_{2000}^{1}+\mu_{2000}^{2}) M
     \label{eq:model1}
\end{align}
where \( \alpha \) denotes the growth rate (the difference between birth rate and death rate) of mongooses on Amami Ōshima Island, while the interaction rates between mongooses and snakes, rats, and natural resources are represented by \( \beta \), \( \delta \), and \( \tau \), respectively. Trapping rate of mongooses using baited cages or snares as well as poisoning method are symbolized by $\mu_{2000}^1$ and $\mu_{2000}^1$.\\
Again, On Amami Ōshima Island in Japan, the snake population follows a dynamic pattern influenced by various ecological interactions. The growth of snakes is driven by their natural reproduction, represented by t. Their numbers further increase as they prey on the abundant rat population, a key food source, However, their survival is not without challenges. Mongooses, introduced to control venomous habu snakes, pose a significant threat, reducing the snake population. These phenomena suggests the following equation
\begin{align}
     \frac{dS}{dt}&=\alpha_1 S+ \xi SR-\sigma SM
     \label{eq:model2}
\end{align}
Here, the symbols $\alpha_1,\xi$ and $\sigma$ signify growth rate of snake, increasing rate of snake due to predation on rats and decreasing rate of snakes due to predation by mongoose respectively.\\

Another important species in the ecosystem is the rat, known for its rapid population growth. In the absence of predators, their numbers rise exponentially, thriving on available resources. However, the balance shifts when snakes and mongooses enter the equation. As snakes hunt them for sustenance and mongooses join the chase, the rat population begins to decline. Mathematically, this relationship is captured by the third equation of the model:
\begin{align}
      \frac{dR}{dt}&= \omega R-\gamma_1 RS -\gamma_2 RM 
      \label{eq:model3}
\end{align}
where \( \omega \) represents the natural growth rate of the rat population, while \( \gamma_1 \) and \( \gamma_2 \) denote the decline in rat numbers due to predation by snakes and mongooses,  correspondingly.\\

Signifiantly, rhe natural resoures of Amami Ōshima Island, including species such as the Amami rabbit, Ryukyu robin, Amami woodcock, and chickens, follows a logistic growth pattern under normal conditions. However, the introduction of mongooses has severely disrupted this balance. As opportunistic predators, mongooses have targeted these native species, leading to a sharp decline in their numbers and threatening the island’s biodiversity. Thus our last equation of our model would be
\begin{align}
      \frac{dN_R}{dt}&= \psi N_R\left(1-\frac{1}{K} N_R\right)-\phi N_R M
      \label{eq:model4}
\end{align}
In this equation, \( \psi \) represents the logistic growth rate of natural resources, \( K \) indicates the carrying capacity of natural resources, and \( \phi \) signifies the depletion rate of resources due to consumption by mongooses.  
\begin{figure}[hbt]
    \centering
\tikzset{every picture/.style={line width=0.75pt}} 
\begin{tikzpicture}[x=0.65pt,y=0.65pt,yscale=-1,xscale=1]
\draw  [color={rgb, 255:red, 209; green, 60; blue, 228 }  ,draw opacity=1 ][fill={rgb, 255:red, 60; green, 255; blue, 172 }  ,fill opacity=1 ] (100.07,1082) -- (152.3,1082) .. controls (159.09,1082) and (164.59,1090.95) .. (164.59,1102) .. controls (164.59,1113.05) and (159.09,1122) .. (152.3,1122) -- (100.07,1122) .. controls (93.28,1122) and (87.78,1113.05) .. (87.78,1102) .. controls (87.78,1090.95) and (93.28,1082) .. (100.07,1082) -- cycle ;
\draw  [color={rgb, 255:red, 209; green, 60; blue, 228 }  ,draw opacity=1 ][fill={rgb, 255:red, 241; green, 255; blue, 154 }  ,fill opacity=1 ] (281.07,982) -- (333.3,982) .. controls (340.09,982) and (345.59,990.95) .. (345.59,1002) .. controls (345.59,1013.05) and (340.09,1022) .. (333.3,1022) -- (281.07,1022) .. controls (274.28,1022) and (268.78,1013.05) .. (268.78,1002) .. controls (268.78,990.95) and (274.28,982) .. (281.07,982) -- cycle ;
\draw  [color={rgb, 255:red, 209; green, 60; blue, 228 }  ,draw opacity=1 ][fill={rgb, 255:red, 238; green, 215; blue, 203 }  ,fill opacity=1 ] (280.07,1170) -- (332.3,1170) .. controls (339.09,1170) and (344.59,1178.95) .. (344.59,1190) .. controls (344.59,1201.05) and (339.09,1210) .. (332.3,1210) -- (280.07,1210) .. controls (273.28,1210) and (267.78,1201.05) .. (267.78,1190) .. controls (267.78,1178.95) and (273.28,1170) .. (280.07,1170) -- cycle ;
\draw  [color={rgb, 255:red, 209; green, 60; blue, 228 }  ,draw opacity=1 ][fill={rgb, 255:red, 235; green, 238; blue, 248 }  ,fill opacity=1 ] (453.07,1079) -- (505.3,1079) .. controls (512.09,1079) and (517.59,1087.95) .. (517.59,1099) .. controls (517.59,1110.05) and (512.09,1119) .. (505.3,1119) -- (453.07,1119) .. controls (446.28,1119) and (440.78,1110.05) .. (440.78,1099) .. controls (440.78,1087.95) and (446.28,1079) .. (453.07,1079) -- cycle ;
\draw [line width=1.5]    (27,1091) -- (84,1091) ;
\draw [shift={(88,1091)}, rotate = 180] [fill={rgb, 255:red, 0; green, 0; blue, 0 }  ][line width=0.08]  [draw opacity=0] (11.61,-5.58) -- (0,0) -- (11.61,5.58) -- cycle    ;
\draw [color={rgb, 255:red, 1; green, 33; blue, 250 }  ,draw opacity=1 ][line width=1.5]  [dash pattern={on 1.69pt off 2.76pt}]  (125,1001) -- (124,1038) ;
\draw [line width=1.5]    (27,1110) -- (84,1110) ;
\draw [shift={(88,1110)}, rotate = 180] [fill={rgb, 255:red, 0; green, 0; blue, 0 }  ][line width=0.08]  [draw opacity=0] (11.61,-5.58) -- (0,0) -- (11.61,5.58) -- cycle    ;
\draw [line width=1.5]    (304,937) -- (304.06,975.5) ;
\draw [shift={(304.07,979.5)}, rotate = 269.91] [fill={rgb, 255:red, 0; green, 0; blue, 0 }  ][line width=0.08]  [draw opacity=0] (11.61,-5.58) -- (0,0) -- (11.61,5.58) -- cycle    ;
\draw [line width=1.5]    (304,1211) -- (304.06,1249.5) ;
\draw [shift={(304.07,1253.5)}, rotate = 269.91] [fill={rgb, 255:red, 0; green, 0; blue, 0 }  ][line width=0.08]  [draw opacity=0] (11.61,-5.58) -- (0,0) -- (11.61,5.58) -- cycle    ;
\draw [line width=1.5]    (483,1163) -- (482.29,1123.25) ;
\draw [shift={(482.22,1119.25)}, rotate = 88.98] [fill={rgb, 255:red, 0; green, 0; blue, 0 }  ][line width=0.08]  [draw opacity=0] (11.61,-5.58) -- (0,0) -- (11.61,5.58) -- cycle    ;
\draw [line width=1.5]    (124,1038) -- (124,1076) ;
\draw [shift={(124,1080)}, rotate = 270] [fill={rgb, 255:red, 0; green, 0; blue, 0 }  ][line width=0.08]  [draw opacity=0] (11.61,-5.58) -- (0,0) -- (11.61,5.58) -- cycle    ;
\draw [line width=1.5]    (267.78,1190) -- (224,1190.92) ;
\draw [shift={(220,1191)}, rotate = 358.8] [fill={rgb, 255:red, 0; green, 0; blue, 0 }  ][line width=0.08]  [draw opacity=0] (11.61,-5.58) -- (0,0) -- (11.61,5.58) -- cycle    ;
\draw [line width=1.5]    (440.78,1099) -- (390,1099.93) ;
\draw [shift={(386,1100)}, rotate = 358.95] [fill={rgb, 255:red, 0; green, 0; blue, 0 }  ][line width=0.08]  [draw opacity=0] (11.61,-5.58) -- (0,0) -- (11.61,5.58) -- cycle    ;
\draw [line width=1.5]    (219.37,1101) -- (168.59,1101.93) ;
\draw [shift={(164.59,1102)}, rotate = 358.95] [fill={rgb, 255:red, 0; green, 0; blue, 0 }  ][line width=0.08]  [draw opacity=0] (11.61,-5.58) -- (0,0) -- (11.61,5.58) -- cycle    ;
\draw [line width=1.5]    (268.78,1002) -- (219,1002) ;
\draw [shift={(215,1002)}, rotate = 360] [fill={rgb, 255:red, 0; green, 0; blue, 0 }  ][line width=0.08]  [draw opacity=0] (11.61,-5.58) -- (0,0) -- (11.61,5.58) -- cycle    ;
\draw [line width=1.5]    (392.37,1189) -- (348.59,1189.92) ;
\draw [shift={(344.59,1190)}, rotate = 358.8] [fill={rgb, 255:red, 0; green, 0; blue, 0 }  ][line width=0.08]  [draw opacity=0] (11.61,-5.58) -- (0,0) -- (11.61,5.58) -- cycle    ;
\draw [color={rgb, 255:red, 1; green, 33; blue, 250 }  ,draw opacity=1 ][line width=1.5]  [dash pattern={on 1.69pt off 2.76pt}]  (215,1002) -- (125,1001) ;
\draw [color={rgb, 255:red, 1; green, 33; blue, 250 }  ,draw opacity=1 ][line width=1.5]  [dash pattern={on 1.69pt off 2.76pt}]  (386,1100) -- (219.37,1101) ;
\draw [line width=1.5]    (395,1002) -- (349.59,1002) ;
\draw [shift={(345.59,1002)}, rotate = 360] [fill={rgb, 255:red, 0; green, 0; blue, 0 }  ][line width=0.08]  [draw opacity=0] (11.61,-5.58) -- (0,0) -- (11.61,5.58) -- cycle    ;
\draw [line width=1.5]    (480,1037) -- (480.54,1075) ;
\draw [shift={(480.6,1079)}, rotate = 269.19] [fill={rgb, 255:red, 0; green, 0; blue, 0 }  ][line width=0.08]  [draw opacity=0] (11.61,-5.58) -- (0,0) -- (11.61,5.58) -- cycle    ;
\draw [color={rgb, 255:red, 1; green, 33; blue, 250 }  ,draw opacity=1 ][line width=1.5]  [dash pattern={on 1.69pt off 2.76pt}]  (480,1002) -- (395,1002) ;
\draw [color={rgb, 255:red, 1; green, 33; blue, 250 }  ,draw opacity=1 ][line width=1.5]  [dash pattern={on 1.69pt off 2.76pt}]  (480.3,1058) -- (480,1002) ;
\draw [line width=1.5]    (123,1166) -- (122.29,1126.25) ;
\draw [shift={(122.22,1122.25)}, rotate = 88.98] [fill={rgb, 255:red, 0; green, 0; blue, 0 }  ][line width=0.08]  [draw opacity=0] (11.61,-5.58) -- (0,0) -- (11.61,5.58) -- cycle    ;
\draw [color={rgb, 255:red, 1; green, 33; blue, 250 }  ,draw opacity=1 ][line width=1.5]  [dash pattern={on 1.69pt off 2.76pt}]  (220,1191) -- (122,1192) ;
\draw [color={rgb, 255:red, 1; green, 33; blue, 250 }  ,draw opacity=1 ][line width=1.5]  [dash pattern={on 1.69pt off 2.76pt}]  (123.15,1194) -- (122.85,1138) ;

\draw (105.17,1087) node [anchor=north west][inner sep=0.75pt]  [font=\large,color={rgb, 255:red, 0; green, 0; blue, 0 }  ,opacity=1 ] [align=left] {$\displaystyle M( t)$};
\draw (286.17,987) node [anchor=north west][inner sep=0.75pt]  [font=\large,color={rgb, 255:red, 0; green, 0; blue, 0 }  ,opacity=1 ] [align=left] {$\displaystyle S( t)$};
\draw (281.17,1176) node [anchor=north west][inner sep=0.75pt]  [font=\large,color={rgb, 255:red, 0; green, 0; blue, 0 }  ,opacity=1 ] [align=left] {$\displaystyle N_{R}( t)$};
\draw (458.17,1084) node [anchor=north west][inner sep=0.75pt]  [font=\large,color={rgb, 255:red, 0; green, 0; blue, 0 }  ,opacity=1 ] [align=left] {$\displaystyle R( t)$};
\draw (33.7,1064) node [anchor=north west][inner sep=0.75pt]   [align=left] {$\displaystyle \alpha M$};
\draw (4.7,1117) node [anchor=north west][inner sep=0.75pt]   [align=left] {$\displaystyle ( u_{1} +u_{2}) M$};
\draw (126,1038) node [anchor=north west][inner sep=0.75pt]   [align=left] {$\displaystyle \beta MS$};
\draw (227,973) node [anchor=north west][inner sep=0.75pt]   [align=left] {$\displaystyle \sigma SM$};
\draw (355,971) node [anchor=north west][inner sep=0.75pt]   [align=left] {$\displaystyle \xi SR$};
\draw (268,930) node [anchor=north west][inner sep=0.75pt]   [align=left] {$\displaystyle \alpha _{1} S$};
\draw (180,1073) node [anchor=north west][inner sep=0.75pt]   [align=left] {$\displaystyle \delta MR$};
\draw (396,1071) node [anchor=north west][inner sep=0.75pt]   [align=left] {$\displaystyle \gamma RM$};
\draw (487,1043) node [anchor=north west][inner sep=0.75pt]   [align=left] {$\displaystyle \gamma RS$};
\draw (450,1142) node [anchor=north west][inner sep=0.75pt]   [align=left] {$\displaystyle \omega R$};
\draw (359,1158) node [anchor=north west][inner sep=0.75pt]   [align=left] {$\displaystyle \psi NR$};
\draw (311,1214) node [anchor=north west][inner sep=0.75pt]   [align=left] {$\displaystyle \frac{\psi }{K} N_{R}^{2}$};
\draw (207,1160) node [anchor=north west][inner sep=0.75pt]   [align=left] {$\displaystyle \phi N_{R} M$};
\draw (127,1133) node [anchor=north west][inner sep=0.75pt]   [align=left] {$\displaystyle \tau MN_{R}$};
\end{tikzpicture}
    \caption{\centering Flow Diagram of the Model \eqref{eq:model}}
    \label{fig:Flow Diagram of the Model }
\end{figure}
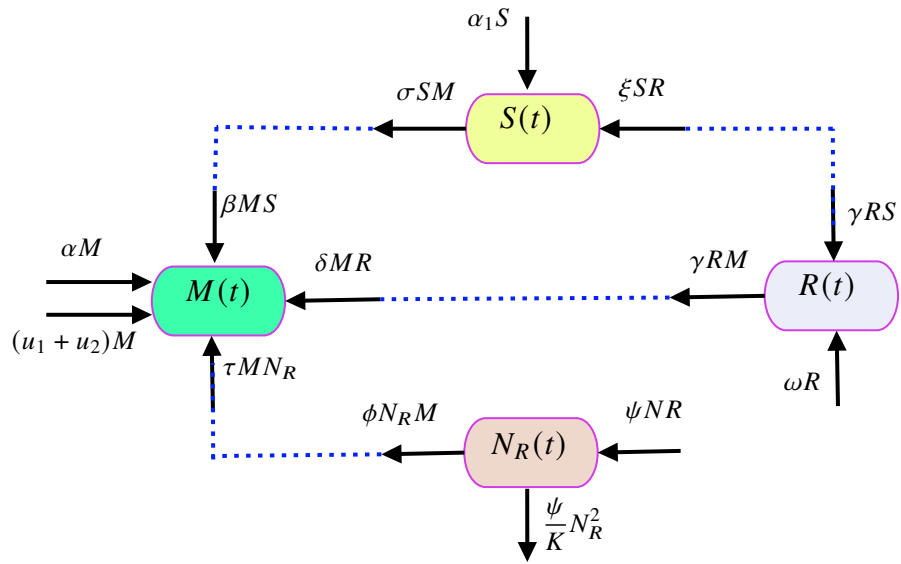
\\ However, combining equations \eqref{eq:model1}-\eqref{eq:model4}, we have obtained our proposed model as follows:
\begin{equation}
  \left\{
    \begin{aligned}
    \frac{dM}{dt}&=\alpha M+\beta MS+\delta MR+\tau MN_R-(\mu_{2000}^{1}+\mu_{2000}^{2}) M \\
    \frac{dS}{dt}&=\alpha_1 S+ \xi SR-\sigma SM  \\
    \frac{dR}{dt}&= \omega R-\gamma_1 RS -\gamma_2 RM \\
    \frac{dN_R}{dt}&= \psi N_R\left(1-\frac{1}{K} N_R\right)-\phi N_R M
\end{aligned} 
  \right.
\label{eq:model}
\end{equation}
with initial conditions $M(0)=M_0\ge 0, S(0)=S_0\ge 0, R(0)=R_0\ge 0$ and $N_R(0)=N_{R_0}
\ge 0$.

\section{Analytical  Analysis}\label{sec:Analytical  Analysis}
With an emphasis on the system's existence, stability at equilibrium points, and different dynamic behaviours, this section provides  analytical analysis of the newly proposed model \eqref{eq:model}.

\subsection{Non-negativity of the Model Solution}\label{Non-negativity of the Model Solution}
According to population-based biological models, populations must always remain non-negative. Consequently, to ensure our model incorporates a positive population, we posit that the initial population size is determined in a way that ensures all population compartments remain positive at all future time points.

\begin{Lemma}
Assuming $M(0)\ge 0, S(0)\ge 0, R(0)\ge 0$ and $N_R(0)\ge0$, then $M(t), S(t), R(t)$ and $N_R(t)$ will always be positive for all $t \epsilon [0,T]$ in ${\mathbb{R}_{+}^4}$ where $T>0$.
\label{Theorem:Positivity Theorem}\index{Positivity}
\end{Lemma}

\begin{proof}
Assume that all of the system's parameters and initial values are positive. Now, we have to prove that $M(t),S(t),R(t)$ and $N_R(t)$ will be positive for all  $t \epsilon [0,T]$ in ${\mathbb{R}_{+}^4}$.\\
From the first equation of the model \eqref{eq:model}, we can write  as follows
\begin{align*} 
\frac{dM}{dt} &=\alpha M+\beta S M+\delta R M +\tau N_R M-(u_{2000}^1+u_{2000}^2) M \\ 
	\Longrightarrow \frac{dM}{dt}  & =\biggl(\alpha +\beta S +\delta R  +\tau N_R -(u_{2000}^1+u_{2000}^2) \biggr)M
 \end{align*}
So, the differential form of the above equation would be
 \begin{align*}
\frac{dM}{M}   =\biggl(\alpha +\beta S +\delta R  +\tau N_R -(u_{2000}^1+u_{2000}^2) \biggr)dt
\end{align*}
After integrating both sides,
\begin{align*}
    \Longrightarrow \ln{M} \ge  (\alpha -u_{2000}^1-u_{2000}^2)t+ \int{(\beta S^*+\delta R^*+\tau N^*)}dt 
\end{align*}    
Then, this result gives
\begin{align*}
M = e^{(\alpha -u_{2000}^1-u_{2000}^2)t+ \int{(\beta S^*+\delta R^*+\tau N^*)}dt} > 0
\end{align*}    
Again, from last equation of the model we can write,
\begin{align}
    \frac{dN_R}{dt}=\left(\psi-\frac{\psi N_R}{K}-\phi M\right) N_R(t)
    \label{eq:differential}
    \end{align}
Thus, the differential form of equation \eqref{eq:differential} shows
    \begin{align*}
         \frac{dN_R}{N_R}=\left(\psi-\frac{\psi N_R}{K}-\phi M\right) dt
\end{align*}
After integration, we get
\begin{align*}
    \therefore N_R =e^{\psi t-\frac{\psi}{K}\int{\left({N_R}-\phi M\right)}dt}>0
\end{align*}
Similary, it is easy to prove that $S(t)>0$ and $R(t)>0$ as $t\xrightarrow{}\infty$.
\end{proof}

\subsection{Boundedness of Model's Solution
}\label{Model Solution's Boundedness}
\begin{Lemma}
	\label{Theorem:Boundedness Theorem2}
	All the solution trajectories  $M(t),S(t),R(t)$ and $N_R(t)$ are bounded.
\end{Lemma}
\begin{proof}
From the last equation of the model \eqref{eq:model}, we get
  \begin{align*}
      \frac{dN_R}{dt}\le \psi N_R\left(1-\frac{N}{N_R}\right)
  \end{align*}
Then by using comparison principle \citep{Rajat2023}, we get
\begin{align*}
     \lim_{t\to\infty} Sup\{N_R(t)\}\le \frac{\psi}{\frac{\psi}{K}}&=K \\
     N_{R_{max}}&=K
\end{align*}
Now, adding second and third equation of the model \eqref{eq:model}, 
\begin{align*}
    \frac{d}{dt}\left(S+R\right)=&\alpha_1 S+\xi SR-\sigma SM+\omega R-\gamma_1 RS-\gamma_2 RM\\
    \implies \frac{d}{dt}\left(S+R\right)\le& \alpha_1 S+\xi SR+\omega R-\gamma_1 SR \\
    \implies \frac{d}{dt}\left(S+R\right)\le& \varepsilon(S+R)+(\xi-\gamma_1)SR \qquad\text{where } \varepsilon=\max\{\alpha_1,\omega\} \\
    \implies \frac{d}{dt}\left(S+R\right)\le &\varepsilon(S+R)+(\xi-\gamma_1)\left\{\left(\frac{S+R}{2}\right)^2-\left(\frac{S-R}{2}\right)^2\right\} \\
    \le & \varepsilon(S+R)+\frac{(\xi-\gamma_1)}{4}\left(S+R\right)^2\\
    \le& \Omega \left\{(S+R)+(S+R)^2\right\} \qquad\text{where } \Omega=\max\left\{\varepsilon,\frac{(\xi-\gamma_1)}{4}\right\}  \\
    \therefore S+R \le & \frac{1}{e^{\Omega t}e^{-c}-1}=L \qquad\text{where c is an arbitrary constant.}
\end{align*}
Thus, according to Lemma \eqref{Theorem:Positivity Theorem}, $S(t),R(t)$ for all $t\in[0, T]$ in $\mathbb{R}_4^+$ and so S(t) and B(t)
are bounded.\\
Again, Since we have found that $N_R,S$ and $R$ all are already bounded, thus from the first equation of the model \eqref{eq:model}, we get
\begin{align*}
    &\frac{dM}{dt}+\left(u_{2000}^1+u_{2000}^2-\alpha-\beta L-\delta L-\tau K\right)M=0 \\
    \implies&\frac{dM}{dt}+GM=0 \qquad\text{where } G=u_{2000}^1+u_{2000}^2-\alpha-\beta L-\delta L-\tau K\\
    \therefore\quad & M(t)=ce^{-Gt} 
\end{align*}
If \( G \) is positive, then as \( t \to \infty \), the exponent \( -Gt \) becomes very large and negative. This causes $e^{-Gt} \to 0$ to decay exponentially toward zero, i.e. $M(t)\le 0$.\\
These all complete our proof.
\end{proof}

\subsection{Existence and Uniqueness of Model's Solution}\label{Existence and Uniqueness of Solution}
\begin{theorem}[Existence and uniqueness of the model solution]
\label{Theorem:(Existence and uniqueness of the model's solution)}
Let $D$ be a domain where the Lipschitz conditions are satisfied. Then, for all non-negative initial conditions, the system's solutions exist and remain unique for all $T\ge 0$ within the domain $D$.
\end{theorem}
\begin{proof}
It has been observed that the Lipschitz condition for the existence and uniqueness of a solution must be adhered to in a region $D$, as suggested by the proposed theorem described in \cite{sowole2009,Kundu2024guava}. Let us consider that
\begin{align*}
	f_1(M, S, R, N_R)&=\alpha M+\beta MS+\delta MR+\tau MN_R-(\mu_{2000}^{1}+\mu_{2000}^{2}) M\\
	f_2(M, S, R, N_R)&=\alpha_1 S+ \xi SR-\sigma SM\\
    f_3(M, S, R, N_R)&=\omega R-\gamma_1 RS -\gamma_2 RM\\
    f_4(M, S, R, N_R)&=\psi N_R\left(1-\frac{1}{K} N_R\right)-\phi N_R M
\end{align*}
By using the given system of equations, the partial derivatives of \( f_1, f_2, f_3 \), and \( f_4 \) with respect to the compartments $M,S,R$ and \( N_R \) are obtained as:  
\begin{align*}
	&\abs{\frac{\partial f_1}{\partial M}}= \abs{\alpha+\beta S+\delta R+\tau N_R-\mu_{2000}^1-\mu_{2000}^2}<\infty,\abs{\frac{\partial f_1}{\partial S}}= \abs{\beta M}=\beta M<\infty,\\
    &\abs{\frac{\partial f_1}{\partial R}}= \abs{\delta M}=\delta M<\infty, \abs{\frac{\partial f_1}{\partial N_R}}= \abs{\tau M}=\tau M<\infty
\end{align*}
Again,
\begin{align*}
	&\abs{\frac{\partial f_2}{\partial M}}= \abs{-\sigma S}=\sigma S<\infty,\abs{\frac{\partial f_2}{\partial S}}= \abs{\alpha_1+\xi R-\sigma M}<\infty,\\
    &\abs{\frac{\partial f_2}{\partial R}}= \abs{\xi S}=\xi S<\infty, \abs{\frac{\partial f_2}{\partial N_R}}= \abs{0}<\infty
\end{align*}
Also,
\begin{align*}
	&\abs{\frac{\partial f_3}{\partial M}}= \abs{-\gamma_2 R}=\gamma_2 R<\infty,\abs{\frac{\partial f_3}{\partial S}}= \abs{-\gamma_1 R}=\gamma_1 R<\infty,\\
    &\abs{\frac{\partial f_3}{\partial R}}= \abs{\omega-\gamma_1 R-\gamma_2 M}<\infty, \abs{\frac{\partial f_3}{\partial N_R}}= \abs{0}<\infty
\end{align*}
Furthermore,
\begin{align*}
	&\abs{\frac{\partial f_4}{\partial M}}= \abs{-\phi N_R}=\phi N_R R<\infty,\abs{\frac{\partial f_4}{\partial S}}= \abs{0}<\infty,\\
    &\abs{\frac{\partial f_4}{\partial R}}= \abs{0}<\infty, \abs{\frac{\partial f_4}{\partial N_R}}= \abs{\psi-2\frac{\psi N_R}{K}}<\infty,
\end{align*}
Therefore, we have demonstrated that all the partial derivatives are continuous and bounded, which ensures that Lipschitz's conditions are met. As a result, by the theorem presented in \citet{sowole2009}, there exists a unique solution to the system \eqref{eq:model} within the domain $D$.

\end{proof}

\subsection{Equilibrium Analysis}\label{se: equailibrium analysis}
Determining the equilibrium points of a differential equation is the initial step in understanding its behavior. In this first stage of analysis, we identify the equilibrium points in this section. To achieve this, we set the system of model equations to zero, i.e.
\begin{align*}
    \frac{dM}{dt}=\frac{dS}{dt}=\frac{dR}{dt}=\frac{dN_R}{dt}=0
\end{align*}
The system's equating to zeros has enabled us to
\begin{align}
    &\alpha M+\beta MS+\delta MR+\tau MN_R-(\mu_{2000}^{1}+\mu_{2000}^{2}) M=0 \label{Eq:Equilibrium equations1}\\
   &\alpha_1 S+ \xi SR-\sigma SM=0\label{Eq:Equilibrium equations2}\\
    &\omega R-\gamma_1 RS -\gamma_2 RM =0 \label{Eq:Equilibrium equations3}\\
    &\psi N_R\left(1-\frac{1}{K} N_R\right)-\phi N_R M=0\label{Eq:Equilibrium equations4} 
\end{align}
After solving the equations (\ref{Eq:Equilibrium equations1}) - (\ref{Eq:Equilibrium equations4}), we get at most four types of following non-negative equilibrium points $E_1, E_2,E_3$ and $E_4$ considering $K_1=\frac{\psi}{K}, \alpha_{11}=\alpha-(\mu_{2000}^{1}+\mu_{2000}^{2})$
\begin{enumerate}
    \item The trivial  equilibrium point $ E_1=\left(0,0,0,0\right) $
    \item Only Natural resources persist equilibrium point 
 $E_2=\left(0,0,0,\frac{\psi}{K_1}\right)$
   \item Snakes and rats free equilibrium point $E_3=\left(\frac{K_1\alpha_{11}+\psi\tau}{\phi\tau},0,0,\frac{\alpha_{11}}{\tau}\right)$
    \item The positive interior equilibrium point  $E_4=\left(M^*,S^*,R^*,N_R^*\right)$ where 
     $ M^*=\frac{\omega}{\gamma_2},\quad S^*=\frac{\alpha_{11}}{\beta},\\ \quad R^*=\frac{K_1 \alpha_{11} \gamma_2 + \tau (\gamma_2 \psi - \omega \phi)}{K_1 \delta \gamma_2}, \quad
    N_R^*=\frac{\gamma_2 \psi - \omega \phi}{K_1 \gamma_2}$ \\
    \text{under the conditions} $\alpha > \left( u_{2000}^{1} + u_{2000}^{2} \right)$, $ \quad \gamma_2 \psi > \omega \phi.$ 
\end{enumerate}

\subsection{Future Status of Natural Resources}
Applying the next-generation matrix method illustrated in \cite{Diekmann1990OnTD}, we determine the future status of natural resources, denoted as $\mathcal{R}_0$.
 To analyze the dynamical behavior of these resources—essential for maintaining ecological balance on Amami Oshima Island—we consider the state variable $N_R$ in our discussion. Let us taken
\begin{equation*}
    \frac{dN_R}{dt}=\psi N_R\left(1-\frac{N}{K}\right)-\phi N_R M
\end{equation*}
Differentiating it with respect to $N_R$, we get
\begin{equation*}
    \Longrightarrow \frac{d}{dN_R}\left(\frac{dN_R}{dt}\right)_{M=M^*,S=S^*,R=R^*,N_R=N_R^*}=\psi-\frac{2\psi N_R^*}{K}-\phi M^*
\end{equation*}
Therefore, 
\begin{align*}
   F= \begin{bmatrix}
       \psi
    \end{bmatrix}_{1\times 1} \text{ and } V= \begin{bmatrix}
      \frac{2\psi N_R^*}{K}+\phi M^*
    \end{bmatrix}_{1\times 1} 
\end{align*}
where $F$ and $V$ represent the increasing and declining of natural resources, respectively.\\ 
Then the future situation of natural resources $N_R$ would be $\mathcal{R}_0$ would be
\begin{equation*}
    \mathcal{R}_0=FV^{-1}=\frac{\psi}{ \frac{2\psi N_R^*}{K}+\phi M^*}=\frac{\psi K}{2\psi N_R^*+\phi K M^*}
\end{equation*}
As a result, natural resources will flourish when \( \psi K \) surpasses \( 2\psi N_R^* + \phi K M^* \). Conversely, if \( \psi K \) falls below \( 2\psi N_R^* + \phi K M^* \), a decline in natural resources will be observed.\\
 However, at the interior equilibrium point, the future status of Natural resources $N_R$ in Amami Oshima can be written as 
\begin{align*}
   \mathcal{R}_0&=\frac{\psi K}{2\psi \frac{K(\gamma_2\psi-\omega\phi)}{\psi\gamma_2}+\phi K \frac{\omega}{\gamma_2}}\\
   \therefore \mathcal{R}_0&=\frac{\psi \gamma_2}{2\psi (\gamma_2\psi-\omega\phi) +\phi  \omega}
\end{align*}

\subsection{Stability Analysis of the Model}
To proof theorems of stability analysis of the model \eqref{eq:model}, at first we consider

\begin{align}
f_1(M^*,S^*,R^*,N_R^*)&=\alpha M^*+\beta M^*S^*+\delta M^*R^*+\tau M^*N_R^*-(\mu_{2000}^{1}+\mu_{2000}^{2}) M^* \label{eq:stability1}  \\
     f_2(M^*,S^*,R^*,N_R^*)& =\alpha_1 S^*+ \xi S^*R^*-\sigma S^*M^* \label{eq:stability2}\\
 f_3(M^*,S^*,R^*,N_R^*)&= \omega R^*-\gamma_1 R^*S^* -\gamma_2 R^*M^* \label{eq:stability3} \\
     f_4(M^*,S^*,R^*,N_R^*)& =\psi N_R^*\left(1-\frac{1}{K} N_R^*\right)-\phi N_R^* M^* \label{eq:stability4}
\end{align}
For the equation (\ref{eq:stability1}) -- (\ref{eq:stability4}), the jacobian matrix is 
\begin{equation}
    J=\frac{\partial(f_1,f_2,f_3,f_4)}{\partial(M^*,S^*,R^*,N_R^*)}=\begin{pmatrix}
\frac{\partial{f_1}}{\partial M^*} & \frac{\partial{f_1}}{\partial S^*}  & \frac{\partial{f_1}}{\partial R^*} & \frac{\partial{f_1}}{\partial N_R^*} \\

\frac{\partial{f_2}}{\partial M^*} & \frac{\partial{f_2}}{\partial S^*}  & \frac{\partial{f_2}}{\partial R^*} & \frac{\partial{f_2}}{\partial N_R^*}\\

\frac{\partial{f_3}}{\partial{S_{f_0}^*}} & \frac{\partial{f_3}}{\partial{S_{f_1}^*}}  & \frac{\partial{f_3}}{\partial{A_f}^*} & \frac{\partial{f_3}}{\partial{S_c}^*} \\

\frac{\partial{f_4}}{\partial{S_{f_0}^*}} & \frac{\partial{f_4}}{\partial{S_{f_1}^*}}  & \frac{\partial{f_4}}{\partial{A_f}^*} & \frac{\partial{f_4}}{\partial{S_c}^*}\\
\end{pmatrix}
\end{equation}

\begin{align}
    \Longrightarrow J= 
\end{align}
\begin{equation}
   \begin{pmatrix}
	\alpha_{11}+\beta S^*+\delta R^*+\tau N^*& \beta M^*  & \delta M^*  & \tau M^*  \\
	-\sigma S^*  & \alpha_1 - \sigma M^* + \xi R^* & \xi S^* &  0 \\
	-\gamma_2 R^* & \gamma_1 R^* & \omega-\gamma_2 M^*-\gamma_1 S^* & 0\\
	-\phi N^*  & 0  & 0 & \psi-\phi M^*-2K_1 N^*
\end{pmatrix}
\label{jacobian}
\end{equation}
where $K_1=\frac{\psi}{K}, \alpha_{11}=\alpha-(u_{2000}^1+u_{2000}^2)$

\begin{theorem}[Local Stability Theorem]
\label{Theorem:Local Stability Theorem1}
The dynamical system \eqref{eq:model} is locally unstable at the trivial equilibrium point $E_1$. 
\end{theorem}                
\begin{proof}
At the equilibrium point $ E_1=\left(0,0,0,0\right)$, the equation (\ref{jacobian}) becomes
\begin{equation}
    \Longrightarrow J_{\arrowvert E_1}=\begin{pmatrix}
\alpha_{11}& 0  & 0 & 0  \\
	0 & \alpha_1& 0 &  0 \\
	0 & 0 & \omega & 0\\
	0 & 0  & 0 & \psi
\end{pmatrix}
\end{equation}
\textbf{So, the characteristic equation is} $\begin{vmatrix}
        J_{\arrowvert E_1}-\lambda I
    \end{vmatrix}=0$
\begin{align}
 \begin{vmatrix}
 	\alpha_{11}-\lambda& 0  & 0 & 0  \\
	0 & \alpha_1-\lambda& 0 &  0 \\
	0 & 0 & \omega-\lambda & 0\\
	0 & 0  & 0 & \psi-\lambda
 \end{vmatrix}=0
\end{align}
From this,it is clear that there are  four different values of eigenvalue $(\lambda)$ and they are
\begin{align}
	\lambda_1 &=\alpha_{11} \\
	\lambda_2 &=\alpha_1 \\
	\lambda_3 &=\omega \\
	\lambda_4 &=\psi  
\end{align}
We see that all of the eigenvalues cannot be negative.  Hence, the model or dynamical system is unstable at the equilibrium point $E_1$ due to the existence of positive eigenvalues $(\lambda)$.\\
Thus, our proof is now completed.
\end{proof}

\subsection{Characteristics of states equilibrium values with respect to \texorpdfstring{$\alpha$}{alpha}}

We will discuss the characterization of the equilibrium values of the population of $M(t),S(t),R(t)$ and $N_R(t)$  based on \cite{MISRA2011128}. From  the system \eqref{eq:model}, we obtain two functions of $M^*, S^*$ and $\alpha$ in below:
\begin{align*}
    f(M^*, S^*,\alpha)&= \alpha M^* - \left( u_{2000}^1+u_{2000}^2\right) M^* + \beta M^*S^* + \frac{M^*\tau(\psi-M^*\phi)}{K_1}\\
     g(M^*, S^*,\alpha)&=\alpha_1 S^* - \sigma M^* S^*
\end{align*}
where $K_1$ ia already defined in Section \ref{se: equailibrium analysis}.
Then we get,
\begin{equation*}
    \therefore \frac{dM^*}{d\alpha}=\frac{\begin{vmatrix}
\frac{\partial{  f(M^*, S^*,\alpha)}}{\partial{S^*}} & \frac{\partial{  f(M^*, S^*,\alpha)}}{\partial{\alpha}} \\
\frac{\partial{  g(M^*, S^*,\alpha)}}{\partial{S^*}} & \frac{\partial{  g(M^*, S^*,\alpha)}}{\partial{\alpha}}
\end{vmatrix}}
{\begin{vmatrix}
\frac{\partial{  f(M^*, S^*,\alpha)}}{\partial{M^*}} & \frac{\partial{  f(M^*, S^*,\alpha)}}{\partial{S^*}} \\
\frac{\partial{  g(M^*, S^*,\alpha)}}{\partial{M^*}} & \frac{\partial{  g(M^*, S^*,\alpha))}}{\partial{S^*}}
\end{vmatrix}}
\end{equation*}
\begin{align*}
\frac{dM^*}{d\alpha} &= \frac{K_1 M^* (\alpha_1 - M \sigma)}
{K_1 q_1 (\alpha - \mu_{2000}^1 - \mu_{2000}^2) + \psi \tau (\alpha_1 - M^* \sigma) + 2 M^{*2} \phi \sigma \tau + K_1 S^* \alpha_1 \beta} \\
&\quad + K_1 M^* \alpha \sigma - K_1 M^* \sigma (\mu_{2000}^1+ \mu_{2000}^2) + 2 M^* \alpha_1 \phi \tau
\end{align*}
Using the previous condition ($\alpha>\mu_{2000}^1+\mu_{2000}^2$) and $\alpha_1>M^*\sigma$, it is followed that both the numerator and the denominator are positive. Consequently, it has been observed that $\frac{dM^*}{d\alpha}>0$ which is graphically visualized in Figure \ref{fig:alpha_variation} (a).\\
Again,
\begin{equation*}
	\therefore \frac{dS^*}{d\alpha}=\frac{\begin{vmatrix}
			\frac{\partial{f(M^*, S^*,\alpha)}}{\partial{\alpha}} & \frac{\partial{  f(M^*, S^*,\alpha)}}{\partial{M^*}} \\
			\frac{\partial{g(M^*, S^*,\alpha)}}{\partial{\alpha}} & \frac{\partial{ g(M^*, S^*,\alpha)}}{\partial{M^*}}
	\end{vmatrix}}
	{{\begin{vmatrix}
\frac{\partial{f(M^*, S^*,\alpha)}}{\partial{M^*}} & \frac{\partial{f(M^*, S^*,\alpha)}}{\partial{S^*}} \\
\frac{\partial{ g(M^*, S^*,\alpha)}}{\partial{M^*}} & \frac{\partial{ g(M^*, S^*,\alpha)}}{\partial{S^*}}
\end{vmatrix}}}
\end{equation*}
\begin{align*}
\frac{dS^*}{d\alpha} &= \frac{-K_1 M^* S^* \sigma}
{K_1 q_1 (\alpha - u_1 - u_2) + \psi \tau (\alpha_1 - M^* \sigma) + 2 M^{*2} \phi \sigma \tau + K_1 S^* \alpha_1 \beta} \\
&\quad + K_1 M^* \alpha \sigma - K_1 M^* \sigma (u_1 + u_2) + 2 M^* \alpha_1 \phi \tau
\end{align*}
From the above discussion, it is clear that the numerator is negative, while the denominator, as per the previous condition, remains positive. This helps us to conclude that, $\frac{dS^*}{d\alpha}<0$ as in Figure \ref{fig:alpha_variation}(b).\\
Similarly, it has been observed that $\frac{dR^*}{d\alpha}<0$ and $\frac{dN_R^*}{d\alpha}<0$ under the previous conditions. Also, Figures \ref{fig:alpha_variation} (c) and (d) have justified this scenario.

\subsection{Nature of rats’ population when snakes’ population as well as mongooses’ population are monotonically decreasing}
This part analyzes rat population dynamics as habu snake and mongoose populations continue to decline. That means snake numbers must steadily decrease and the mongoose population must remain lower than before. Thus, the third equation of the model (\ref{eq:model}) suggests that 
\begin{align*}
   \frac{dR}{dt}= \omega R-\gamma_1 RS -\gamma_2 RM
\end{align*}
Now, differentiating both sides with respect to time $(t)$, we get
\begin{align*}
  \frac{d^2R}{dt^2} =& \omega \frac{dR}{dt}-\gamma_1 R \frac{dS}{dt}-\gamma_1 S\frac{dR}{dt}-\gamma_2 R \frac{dM}{dt}-\gamma_1 M\frac{dR}{dt}
  \\ 
  >& \omega \frac{dR}{dt}-\gamma_1 \frac{dR}{dt}-\gamma_2 \frac{dR}{dt}
\end{align*}
Because of the fact that both the population of snakes and mongooses are ontinously delining, so $\frac{dS}{dt}$, $\frac{dM}{dt}$ must be negative and consequently $-\gamma_1 R \frac{dS}{dt}-\gamma_2 R \frac{dM}{dt}$ must be positive term.\\ 
Therefore, at the equilibrium point, it has been observed that
\begin{equation*}
\frac{d^2R}{dt^2}>0
\end{equation*}
Hence, the rat population is minimal at the critical points.

\subsection{Sensitivity Analysis}
We have conducted a sensitivity analysis to evaluate the model's resilience to parameter alterations. This assists in identifying the factors that substantially affect the natural resources delineated in the model's equations. The sensitivity analysis depends on the methodology shown by \cite{omoloye2021mathematical}, using the normalized forward sensitivity index of a variable in relation to a parameter. Furthermore, when the variable is a differentiable function of the parameter, the sensitivity index may be determined by partial derivatives.

\subsubsection{Local Sensitivity Indices for \texorpdfstring{$\mathcal{R}_0$}{R0}}
Definition: The normalized forward sensitivity index of a variable $S$ that depends deferentially on a parameter $\nu$, is defined as
\begin{equation*}
    \mathcal{L}_{\nu}^S=\frac{\partial S}{\partial \nu}\times \frac{\nu}{S}
\end{equation*}
It is expected that the sensitivity indices should vary from -1 to 1. In particular, sensitivity indices of the future status of natural resources $\mathcal{R}_0$, with respect to the model parameter $\nu$ have been calculated using the formula $\frac{\partial \mathcal{R}_0}{\partial \nu} \times \frac{\nu}{\mathcal{R}_0}$.

\begin{table}[hbt]
 \caption{Sensitivity Analysis for the model \eqref{eq:model} }
    \centering
    \begin{tabular}{c c}
     \hline
\textbf{\textcolor{blue}{Parameter}} & \textbf{\textcolor{blue}{Sensitivity index}} \\
\hline 
$\psi$~            &  0.5037               \\
$\omega$          & -0.6775                 \\
$\gamma_2$             & 0.6775                 \\
$\phi$         & -0.6775                  \\
\hline 
    \end{tabular}
    \label{tab:my_label}
\end{table}
\begin{figure}[hbt]
	\centering
\includegraphics[height=6cm,width=9cm]{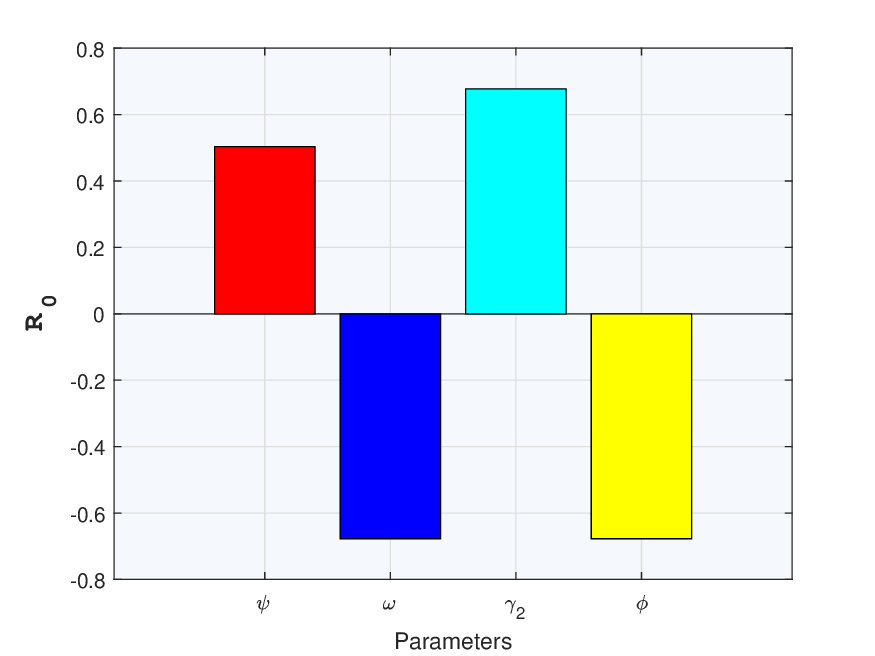}
	\caption{\centering Effects of Parameter Variationson $\mathcal{R}_0$}
	\label{fig:sensitivity analysis}
\end{figure}
After doing all the necessary calculations using the values of Table \ref{tab:parameter_introdution}, we obtain the
sensitivity indices of the parameters whose values have been depicted in Table \ref{tab:my_label} and its graphical view has been illustrated in Figure \ref{fig:sensitivity analysis}, from where we observe that the parameters $\psi$  and $\gamma_2$  are with positive sensitivity
values (need to rise) and the parameters $\omega$  and $\phi$  are with negative sensitivity values (need to decrease). That meanis, to enhance natural resources on the island, it is essential to regulate species interactions that impact resource availability. Increasing the logistic growth rate of natural resources ($\psi$) promotes regeneration, ensuring a sustainable supply. Additionally, boosting mongoose predation on rats ($\gamma_2$) helps reduce the rat population, which alleviates competition for resources and supports natural replenishment. Conversely, lowering the natural growth rate of mice ($\omega$) is beneficial, as a smaller mouse population reduces resource consumption and prevents depletion. Similarly, decreasing the rate at which mongooses consume natural resources ($\phi$) helps maintain their availability. Based on this result, it has been shown that Amami Ōshima Island's natural resource balance is greatly dependent on the existence of mongooses. In their position as rodent controllers, mongooses prevent the overconsumption of natural resources by rats, which may cause ecological imbalance and depletion. Thus, controlled mongoose populations help keep the island's natural resources viable, but only under strict management to avoid over-predation or other undesirable ecological effects.

\section{Numerical Simulations}\label{sec:Model Simulations}
In this section, we dive into a captivating numerical simulation of the proposed model \eqref{eq:model}, executed using MATLAB (R2018b) software. Our goal is to explore the  dynamics of the prey-predator system, intertwined with natural resources, on the picturesque Amami Oshima Island, Japan. To achieve this, we have applied well known fourth-order Runge-Kutta method with using the paramters value mentioned in Table \ref{tab:parameter_introdution}. The simulation spans from 1973 to 2025, capturing the island’s ecological shifts following the introduction of mongooses. The initial values for the variables are set to $M_0 = 30$, $S_0 = 50,000$, $R_0 = 24,240$, and $N_R=80,000$, providing a foundation for this ecological investigation. Since the literature indicates that mongooses are less active at night when habu snakes have come out in the island, the killing rate of habu snakes by mongooses is considered to be relatively low in this simulation prcoess. For finding the numerial simulations, the terms \( u_{2000}^{1} \) and \( \mu_{2000}^{2} \) have been used as active parameters, with their values set from the year 2000.\\
\begin{table}[hbt]
    \centering
    \caption{Model Parameters and Their Values}    
     \begin{tabular}{>{\centering\arraybackslash}p{0.8cm} >{\centering\arraybackslash}p{7.7cm} >{\centering\arraybackslash}p{1.7cm}} 
        \hline
        \multicolumn{1}{c}{\textbf{\textcolor{blue}{Symbol}}} & \multicolumn{1}{c}{\textbf{\textcolor{blue}{Description of Parameter}}} & \multicolumn{1}{c}{\textbf{\textcolor{blue}{Value}}} \\
        \hline
        $\alpha$ & Growth rate of mongooses & 0.294 \\
        $\beta$ & Increasing rate of mongooses due to predation on snake & $8.89 \times 10^{-4}$ \\
        $\delta$ & Rising rate of mongooses due to predation on rat & 0.333 \\
        $\tau$ & Increasing rate of mongooses due to consumption of natural resources & $8.89 \times 10^{-4}$ \\
        $\mu_{2000}^{1}$ & Trapping rate of mongooses using baited cages or snares & 0.5 \\
        $\mu_{2000}^{2}$ & Trapping rate of mongooses due to poisoning method & 0.2 \\
        $\alpha_1$ & Growth rate of snake & 0.063 \\
        $\xi$ & Increasing rate of snake due to predation on rats & $6.24 \times 10^{-8}$ \\
        $\sigma$ & Decreasing rate of snakes due to predation by mongooses & $4.07 \times 10^{-5}$ \\
        $\omega$ & Natural growth rate of the mouse population & 0.0387 \\
        $\gamma_1$ & Decreasing rate of rats due to predation by snakes & $3.82 \times 10^{-9}$ \\
        $\gamma_2$ & Reduction rate of rats due to predation by mongooses & $4.56 \times 10^{-5}$ \\
        $\psi$ & Logistic growth rate of natural resources & 0.09 \\
        $K$ & Carrying capacity of natural resources & 800,000 \\
        $\phi$ & Decreasing rate of resources due to consumption by mongooses & $4.89 \times 10^{-5}$\\
        \hline
    \end{tabular}
     \label{tab:parameter_introdution}
\end{table}
\begin{figure}[hbt]
    \centering
 \begin{subfigure}{0.32\textwidth}
        \centering
        \includegraphics[width=\textwidth]{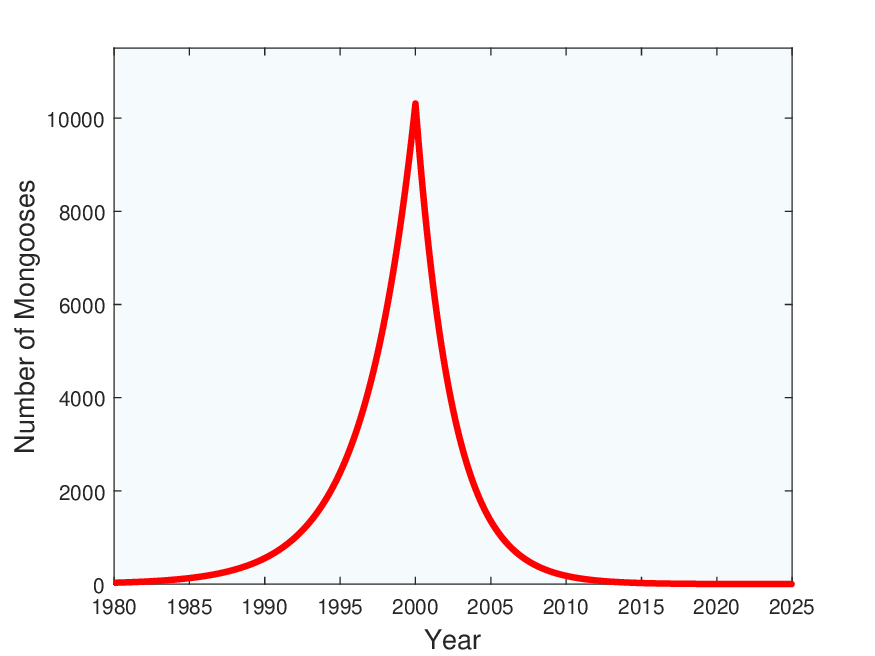}
        \caption{}
    \end{subfigure}
    \hfill
    \begin{subfigure}{0.32\textwidth}
        \centering
        \includegraphics[width=\textwidth]{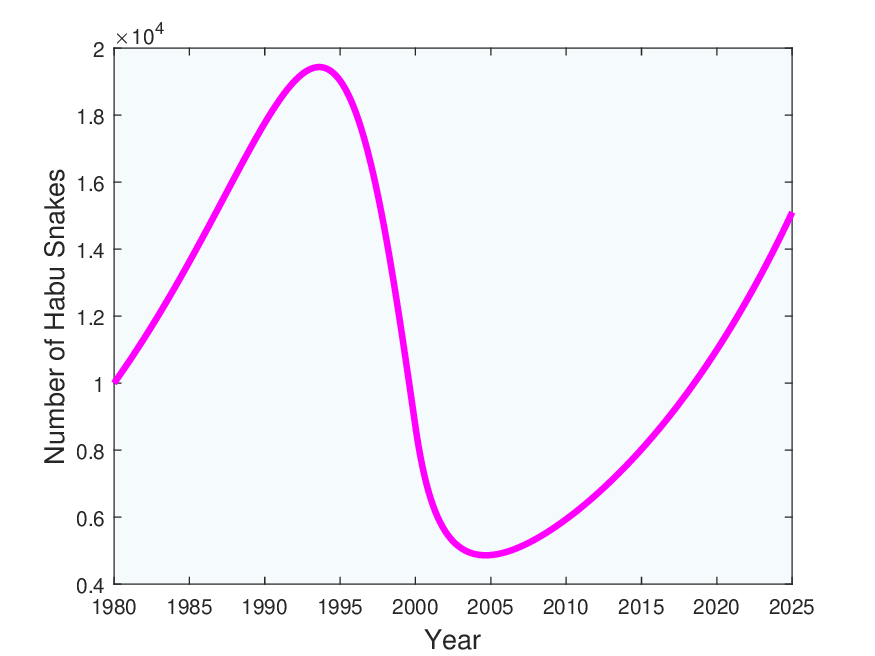}
        \caption{}
    \end{subfigure}
    \hfill
    \begin{subfigure}{0.32\textwidth}
        \centering
        \includegraphics[width=\textwidth]{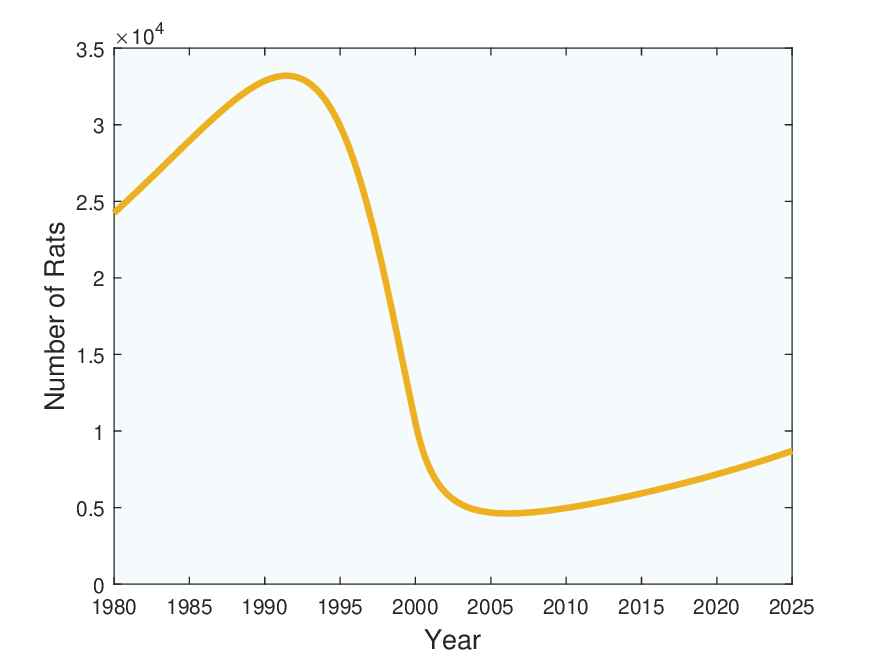}
        \caption{}
    \end{subfigure}
    \hfill
    \begin{subfigure}{0.33\textwidth}
        \centering
        \includegraphics[width=\textwidth]{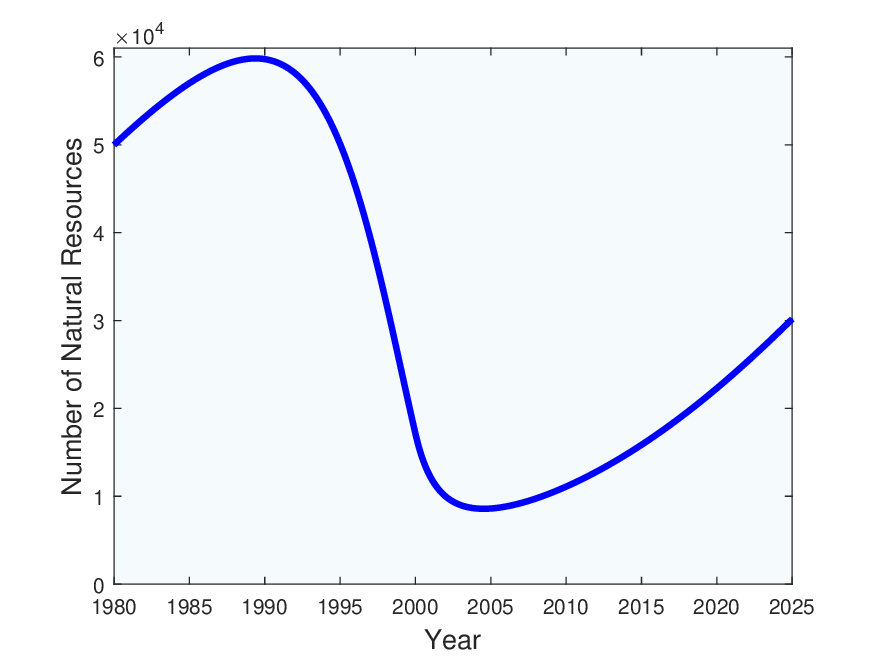}
        \caption{}
    \end{subfigure}
    \hfill
    \begin{subfigure}{0.32\textwidth}
        \centering
        \includegraphics[width=\textwidth]{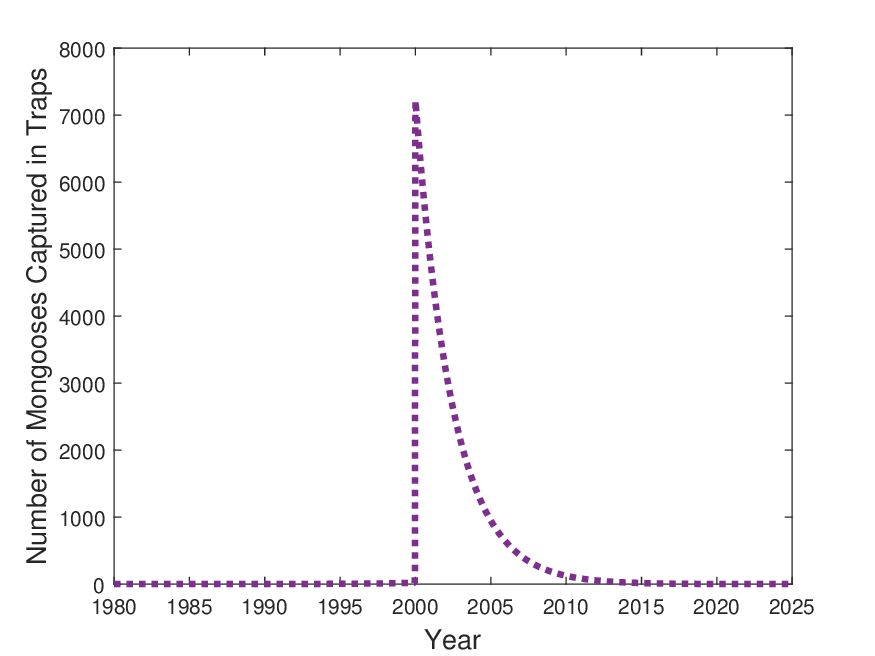}
        \caption{}
        \label{fig:fig44}
    \end{subfigure}
    \hfill
    \begin{subfigure}{0.32\textwidth}
        \centering
        \includegraphics[width=\textwidth]{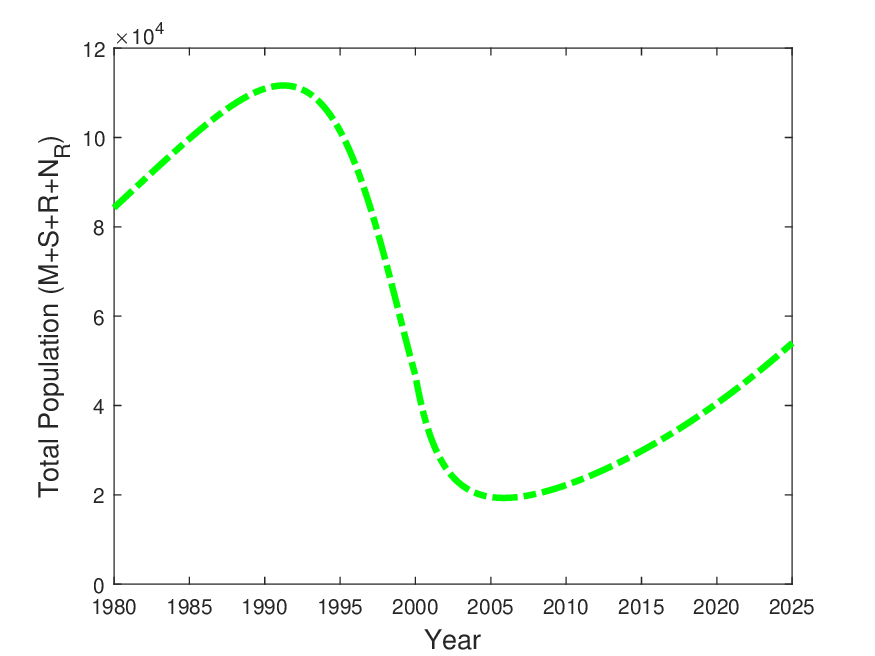}
        \caption{}
    \end{subfigure}
    \caption{Solution trajectories of state  variables: (a) Mongooses (b) Habu Snakes (c) Rats and (d) Natural resources (e) trapping, and (f) total population in Amami Osami Island over the time period from 1980 to 2025 (real-world phenomenon)}
         \label{fig:solution}
\end{figure}
\\
Figure \ref{fig:solution} presents the solution trajectories (time-series) for the proposed model \eqref{eq:model} which describe the real world phenomena observed in Amami Ōshima Island, Japan. The dynamics of mongooses, habu snakes, rats, natural resources, the number of mongooses captured in traps (which includes removal via baited cages or snares and poisoning), and the total population (the combined count of mongooses, habu snakes, rats, and natural resources on the island) are depicted in Figures \ref{fig:solution} (a), (b), (c), (d), (e), and (f), respectively. The first figure reveal that, after the introduction of 30 mongooses in 1979, their population gradually increased. This rise in numbers accelerated significantly starting from 1990, and by 2000, the population quickly reached approximately 10,000 within just 20 years. This trend aligns with the actual mongoose population dynamics on the island, as outlined in \citep{YamadaSugimura2004}. As like as the decision made by the authorities in Japan for capturing mongooses began in 2000, leading to a gradual decline in their population from 2000, which reached zero by 2018. As a result, there was no mongoose population from 2018 to 2025 (see Figure \ref{fig:solution}(a)). During this period, the annual number of mongooses captured is depicted in Figure \ref{fig:solution}(e), starting from the initiation of the capture process in 2000. Besides, Figure \ref{fig:solution}(b) shows that around 1995, the habu snake population increased due to the initially low presence of habu snakes. However, over time, their numbers declined until 2000 due to interactions with the rapidly growing mongoose population. On the other hand, as the mongoose population began to decline due to trapping from 2000 onwards, the habu snake population started to rise again. A similar pattern is observed in Figure \ref{fig:solution}(c) for the rat population on the island. Additionally, Figure \ref{fig:solution} illustrates that natural resources started to decline around 1990, a trend that continued until approximately 2005, reaching a reduced level of about $1 \times 10^4$. After this period, due to the decreased presence of mongooses, natural resources began to recover, though the increase was not substantial. The impact of mongooses led to an almost 50\% reduction in natural resources on Amami Ōshima Island. Furthermore, Figure \ref{fig:solution}(f) indicates that the total population significantly declined after 1990, although a slight increase was observed in the first ten years. By the end, the overall population was lower than the initial population of different species in that island. 
\begin{figure}[hbt]
    \centering
    \begin{subfigure}[b]{0.49\textwidth}
        \centering
        \includegraphics[width=0.8\textwidth]{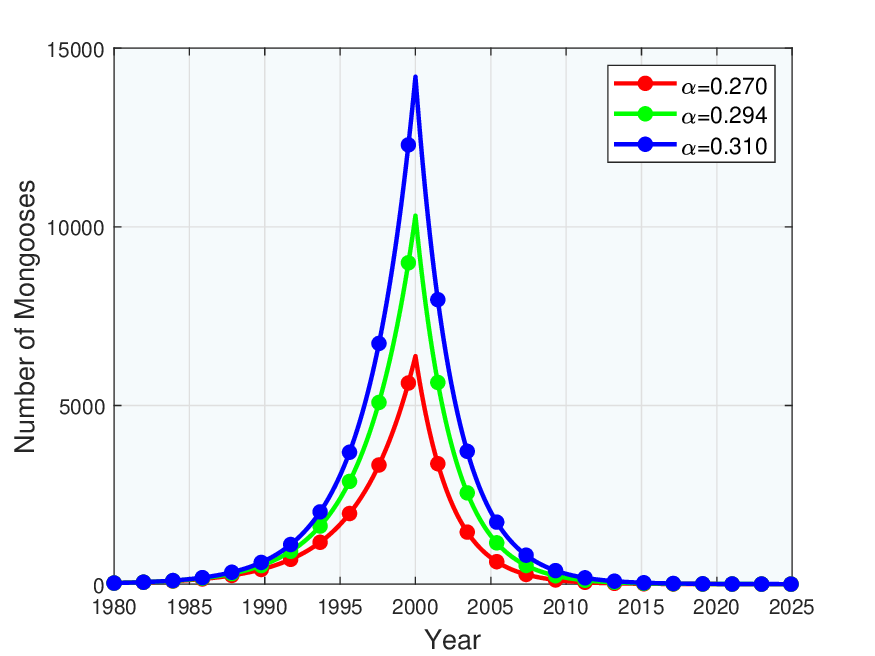}
        \caption{}
    \end{subfigure}
    \hspace{0.001\textwidth} 
    \begin{subfigure}[b]{0.49\textwidth}
        \centering
        \includegraphics[width=0.8\textwidth]{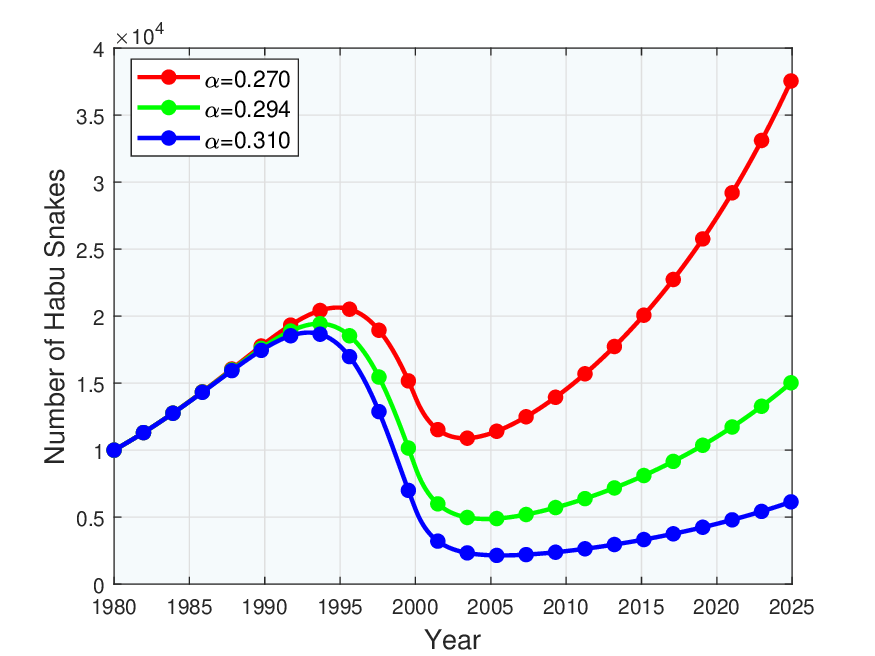}
        \caption{}
    \end{subfigure}
    \hspace{0.0001\textwidth} 
    \begin{subfigure}[b]{0.49\textwidth}
        \centering
        \includegraphics[width=0.8\textwidth]{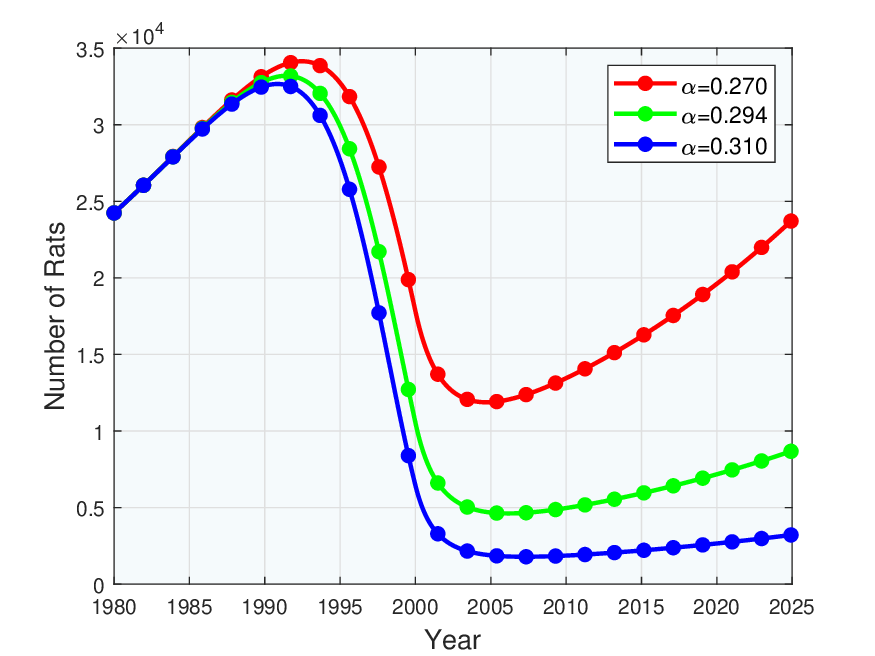}
        \caption{}
    \end{subfigure}
     \hspace{0.001\textwidth} 
    \begin{subfigure}[b]{0.49\textwidth}
        \centering
        \includegraphics[width=0.8\textwidth]{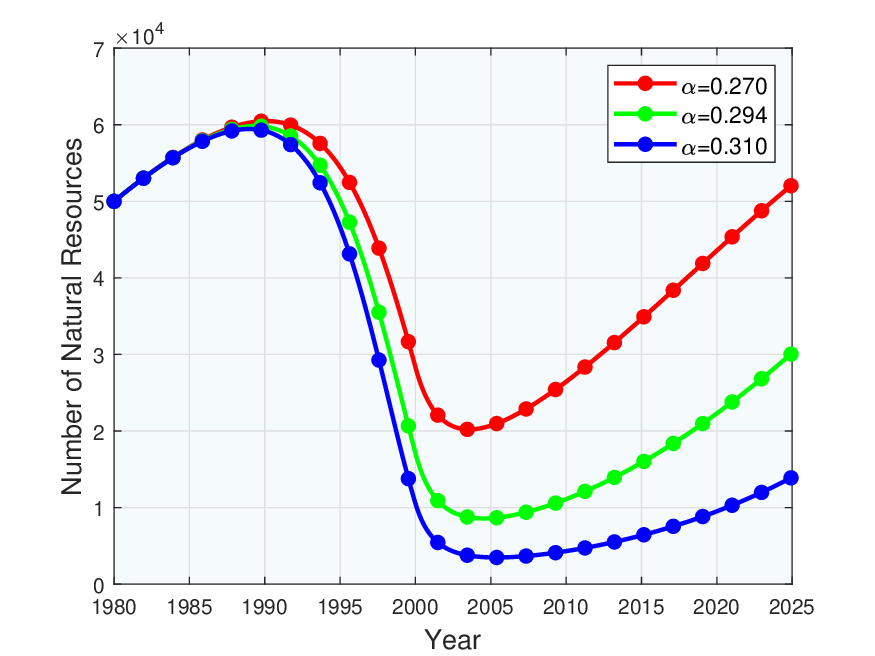}
        \caption{}
    \end{subfigure}
    \caption{(a) As the growth rate of mongooses increases, their population expands proportionally in the ecosystem. (b) As the mongoose growth rate increases, the habu snake population initially grows, but eventually declines due to larger interactions before rising as the mongoose population decreases. (c) The rat population grows faster with increasing mongoose numbers, then declines due to intensified interactions, but later recovers as the mongoose population decreases.(d) A higher mongoose growth rate creates a greater burden on the natural resources of Amami Oshima Island, Japan.}
    \label{fig:alpha_variation}
\end{figure}
\\
In addition, variations in the state variables about $\alpha$ are shown in Figure \ref{fig:alpha_variation}, where $\alpha$ denotes the growth rate of mongooses. Figure \ref{fig:alpha_variation} (a) demonstrates that the mongoose population initially increases for all values of $\alpha$, reaching its peak, with the highest peak observed at $\alpha = 0.310$. Following this peak, the population gradually declines due to trapping. Interestingly, the rate of decline is less pronounced for higher values of $\alpha$, suggesting a potential buffering effect. On the other side, Figure \ref{fig:alpha_variation} (b) reveals that, up until around 1990, there was no significant difference in the habu snake population across different values of $\alpha$. However, beyond this point, variations became more evident. The habu snake population exhibited a gradual decline, reaching its lowest levels around 2003, 2005, and 2006 for $\alpha = 0.270$, $\alpha = 0.294$, and $\alpha = 0.310$, respectively. Notably, after this period, the population rebounded, with a dramatic increase observed by 2025, particularly for lower values of $\alpha$. Besides, a similar dynamic trend is evident in the rat population and natural resource availability, as depicted in Figure \ref{fig:alpha_variation}, further reinforcing the intricate interplay between species and ecological balance on Amami Ōshima Island.\\
Again,  Figure \ref{fig:phase diagram} presents phase portraits illustrating the relationship between state variables. As shown in Figure \ref{fig:phase diagram} (a), the habu snake population initially declines with the rise of the mongoose population.  However, as the mongoose population begins to decrease, the habu snake population starts to recover. Notably, as the mongoose population approaches extinction, the habu snake population exhibits a steady and continuous increase. 
\begin{figure}[hbt]
    \centering
    \begin{subfigure}[b]{0.49\textwidth}
        \centering
        \includegraphics[width=0.8\textwidth]{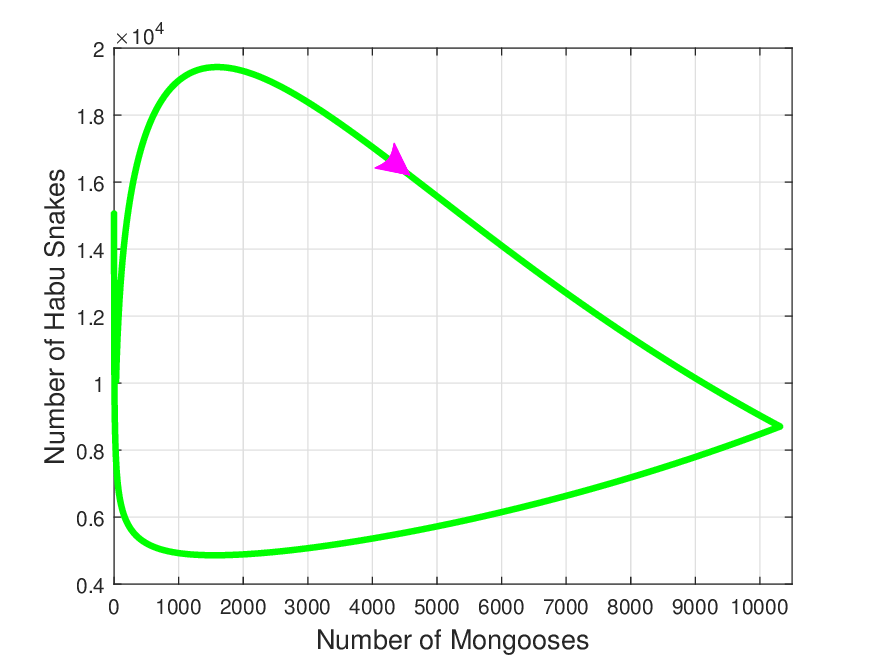}
        \caption{}
    \end{subfigure}
    \hspace{0.001\textwidth} 
    \begin{subfigure}[b]{0.49\textwidth}
        \centering
        \includegraphics[width=0.8\textwidth]{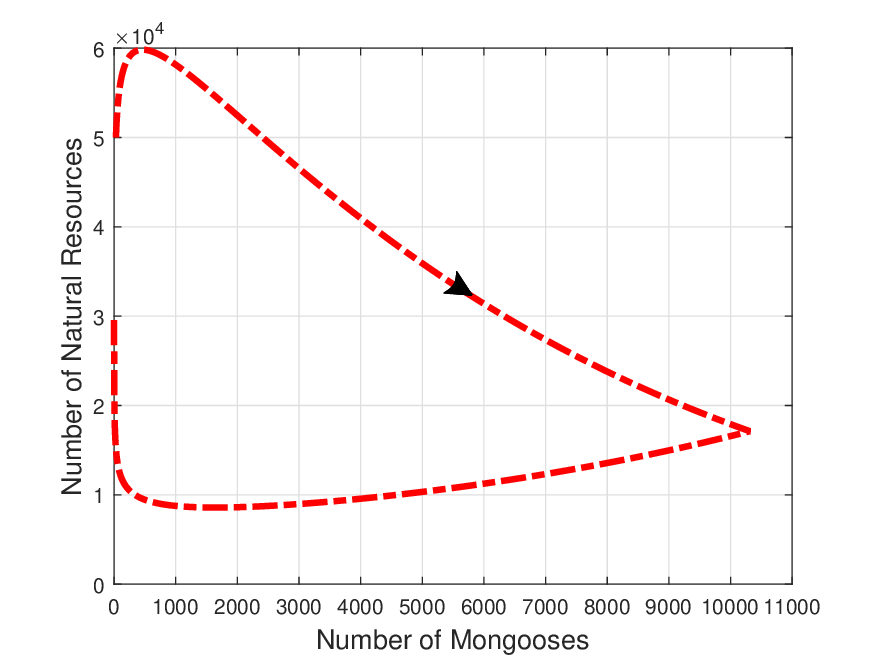}
        \caption{}
    \end{subfigure}
    \hspace{0.0001\textwidth} 
    \begin{subfigure}[b]{0.49\textwidth}
        \centering
        \includegraphics[width=0.8\textwidth]{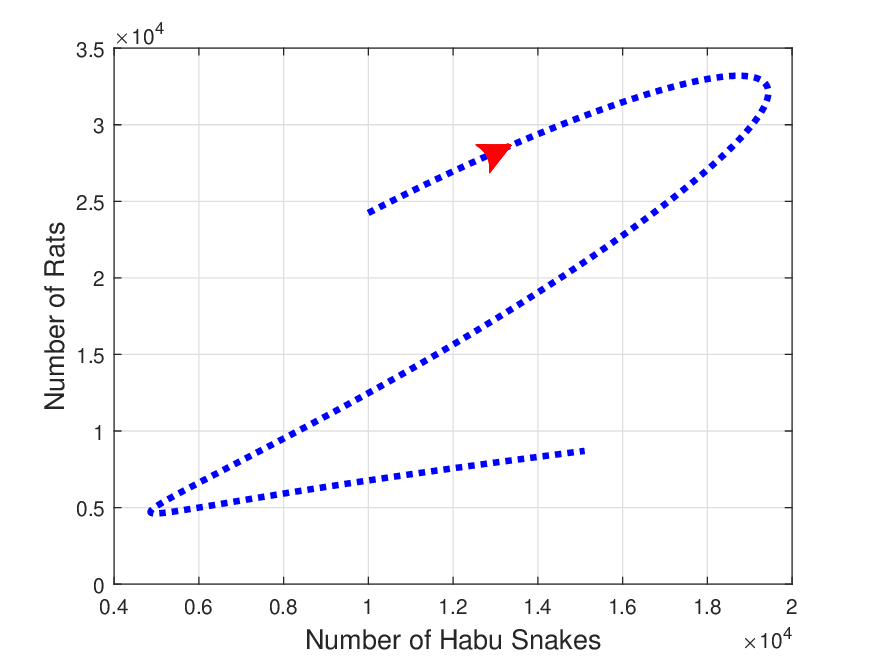}
        \caption{}
    \end{subfigure}
    \caption{(a) Impact of fluctuations in mongoose population on the habu snake population on Amami Oshima Island, Japan (b) Variation in the number of natural resources on Amami Oshima Island with changes in the mongoose population (c) Relationship of the number of rats and number of habu snakes in the island.}
    \label{fig:phase diagram}
\end{figure}

\subsection{Comparison of Various Management Scenarios in Amami Ōshima Island, Japan}
On Amami Ōshima Island, authorities initiated a mongoose eradication program in 2000 to mitigate ecological damage. This effort involved capturing approximately nine mongooses per day using nets and eliminating six through poisoning, continuing until 2018. However, despite these measures, the ecological disruption was already severe and continues to affect the island's biodiversity today. Had the trapping program been implemented differently, with adjusted removal rates, the outcome might have been significantly different. This subsection explores various strategic approaches that could be adopted to prevent ecological degradation, referred to as ``Scenarios" for conservation and ecosystem restoration.
\begin{figure}[hbt]
    \centering
 \begin{subfigure}{0.32\textwidth}
        \centering
        \includegraphics[width=\textwidth]{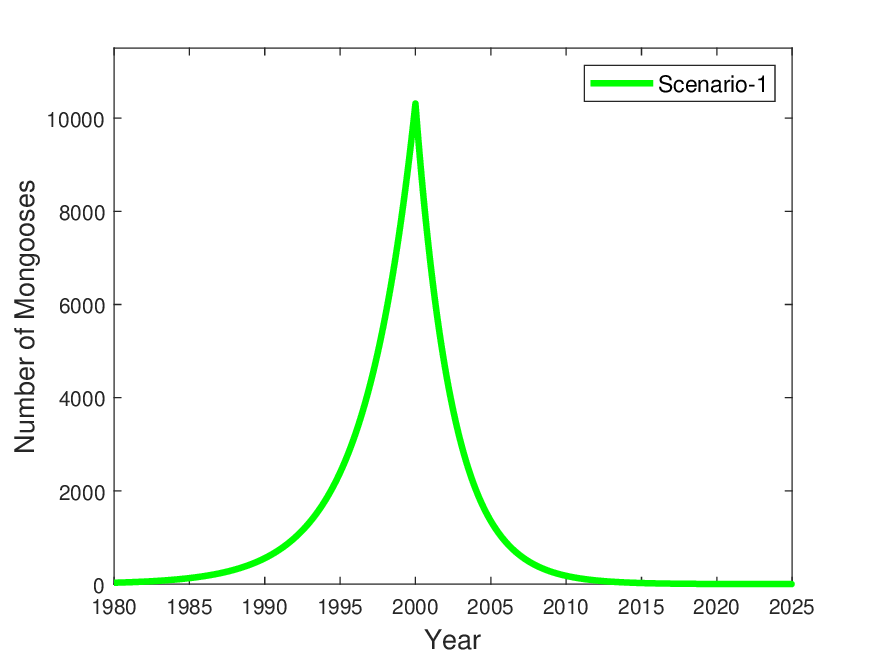}
        \caption{}
    \end{subfigure}
    \hfill
    \begin{subfigure}{0.32\textwidth}
        \centering
        \includegraphics[width=\textwidth]{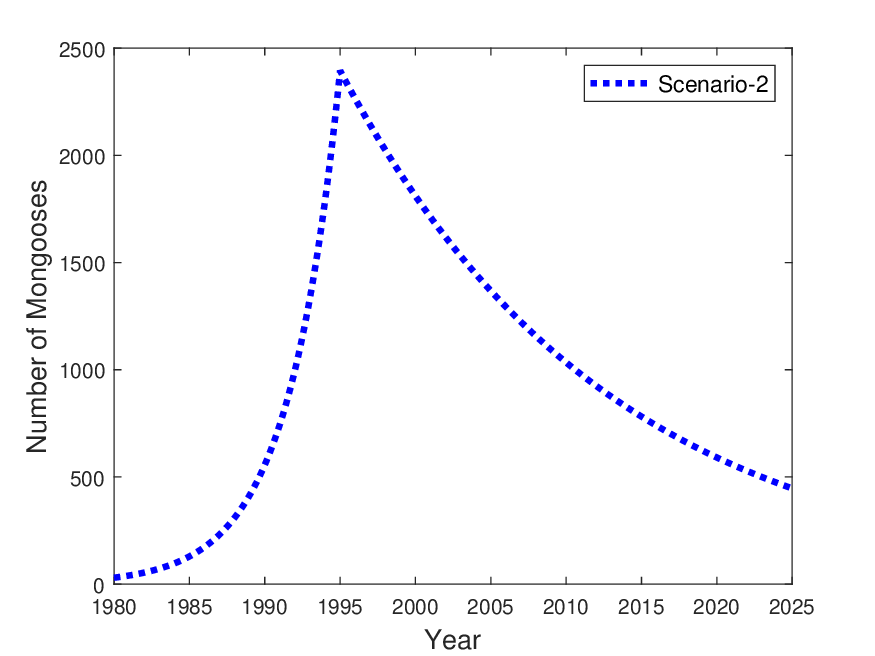}
        \caption{}
    \end{subfigure}
    \hfill
    \begin{subfigure}{0.32\textwidth}
        \centering
        \includegraphics[width=\textwidth]{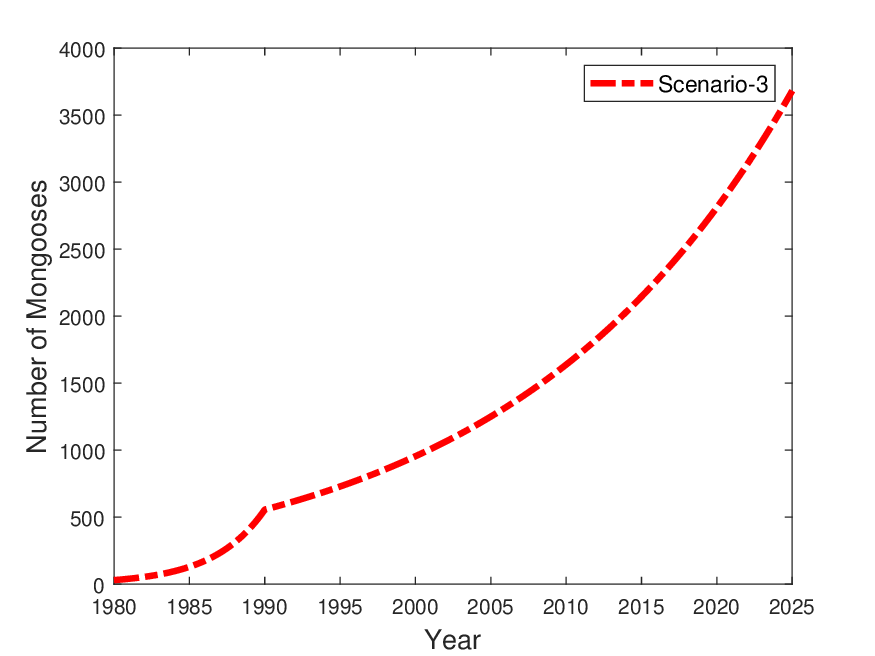}
        \caption{}
    \end{subfigure}
    \hfill
    \begin{subfigure}{0.33\textwidth}
        \centering
        \includegraphics[width=\textwidth]{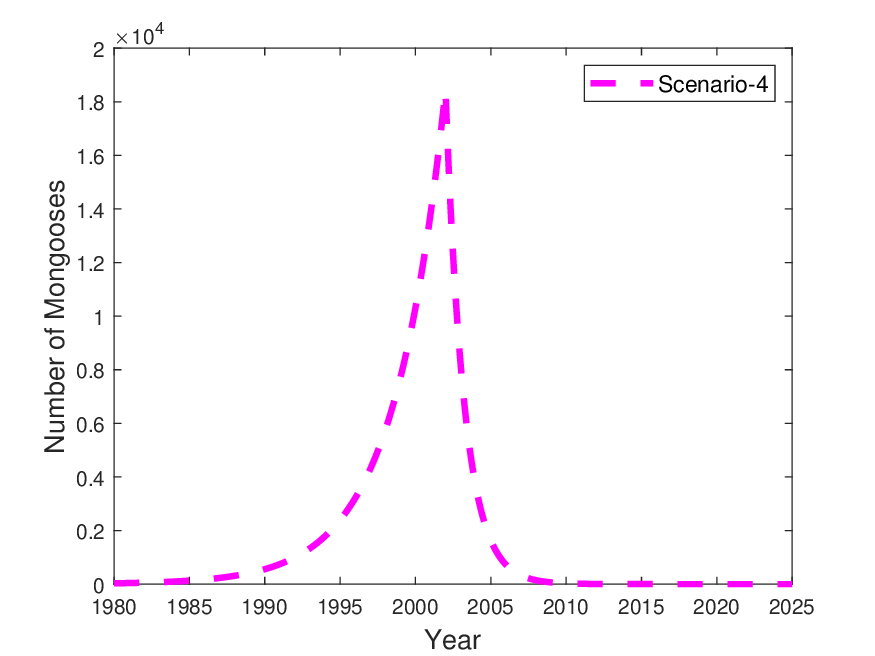}
        \caption{}
    \end{subfigure}
    \hfill
    \begin{subfigure}{0.32\textwidth}
        \centering
        \includegraphics[width=\textwidth]{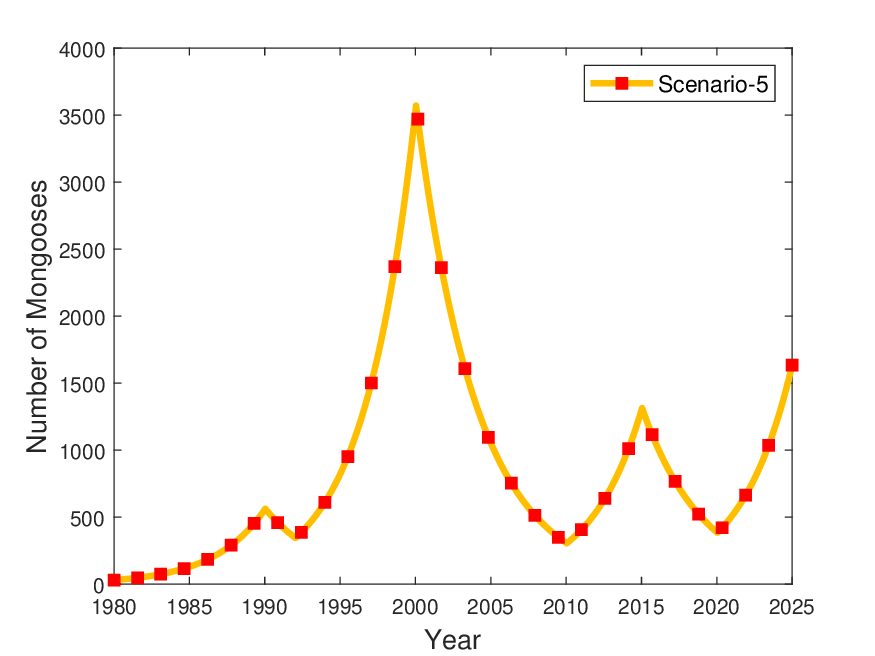}
        \caption{}
    \end{subfigure}
    \hfill
    \begin{subfigure}{0.32\textwidth}
        \centering
        \includegraphics[width=\textwidth]{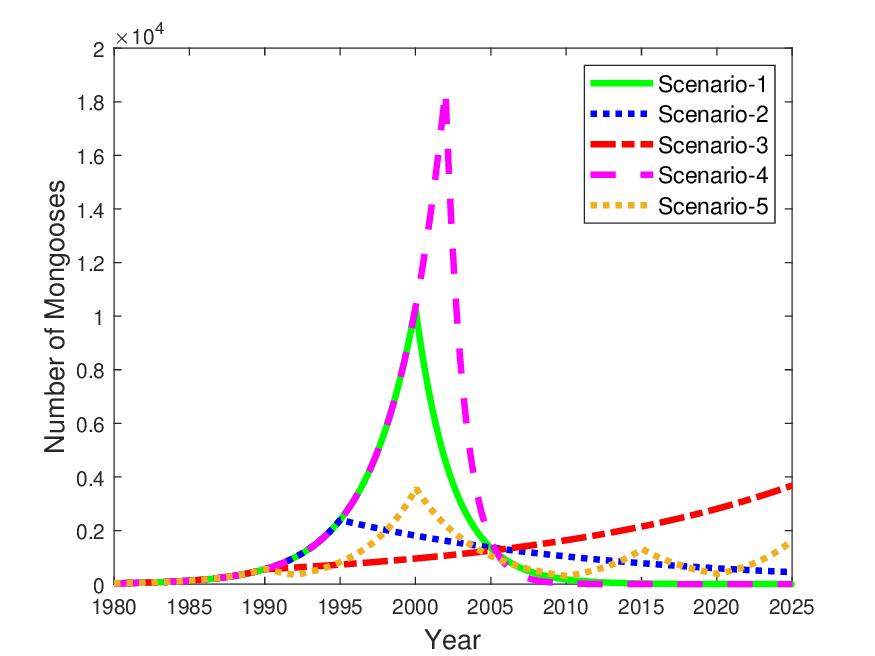}
        \caption{}
    \end{subfigure}
    \caption{The dynamics of the mongoose population over time under different trapping scenarios of mongoose: (a) Scenario-1: Mongoose trapping began in 2000, as implemented on Amami Ōshima Island, reflecting the current trend; (b) Scenario-2: Trapping started in 1995 at a lower rate than in Scenario-1; (c) Scenario-3: Trapping was initiated in 1990 at a lower rate than in Scenario- 1 and 2; (d) Scenario-4: Trapping commenced in 2002 at a higher rate than in Scenario-1; (e) Scenario-5: Trapping occurred in three phases: 1990–1992, 2000–2010, and 2015–2020; (f) A comparison of all scenarios in a single frame }
    \label{fig:Mongooses}
\end{figure}
\begin{figure}[hbt]
    \centering
 \begin{subfigure}{0.32\textwidth}
        \centering
        \includegraphics[width=\textwidth]{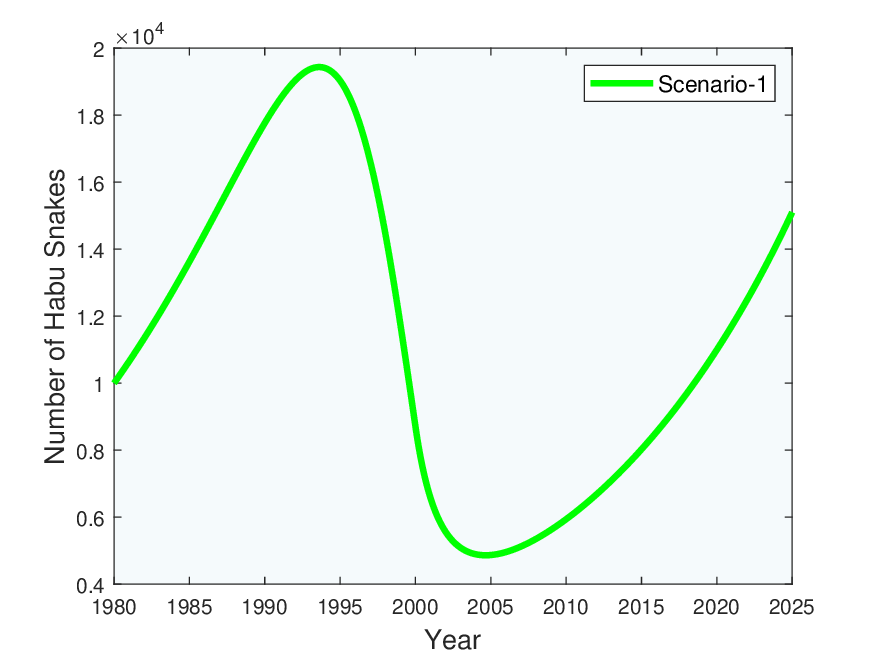}
        \caption{}
        \label{fig:fig333}
    \end{subfigure}
    \hfill
    \begin{subfigure}{0.32\textwidth}
        \centering
        \includegraphics[width=\textwidth]{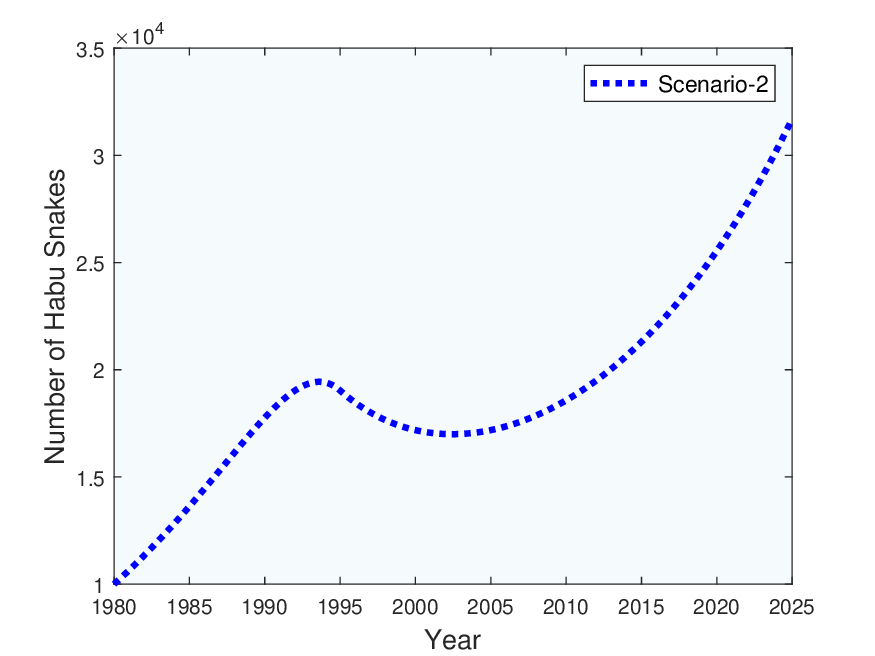}
        \caption{}
        \label{fig:fig444}
    \end{subfigure}
    \hfill
    \begin{subfigure}{0.32\textwidth}
        \centering
        \includegraphics[width=\textwidth]{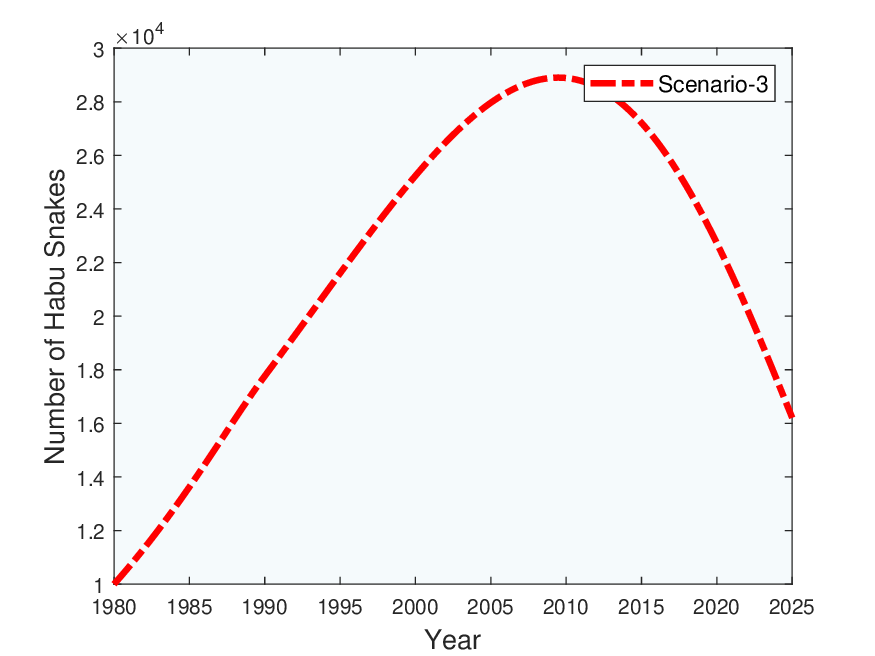}
        \caption{}
        \label{fig:fig5555}
    \end{subfigure}
    \hfill
    \begin{subfigure}{0.33\textwidth}
        \centering
        \includegraphics[width=\textwidth]{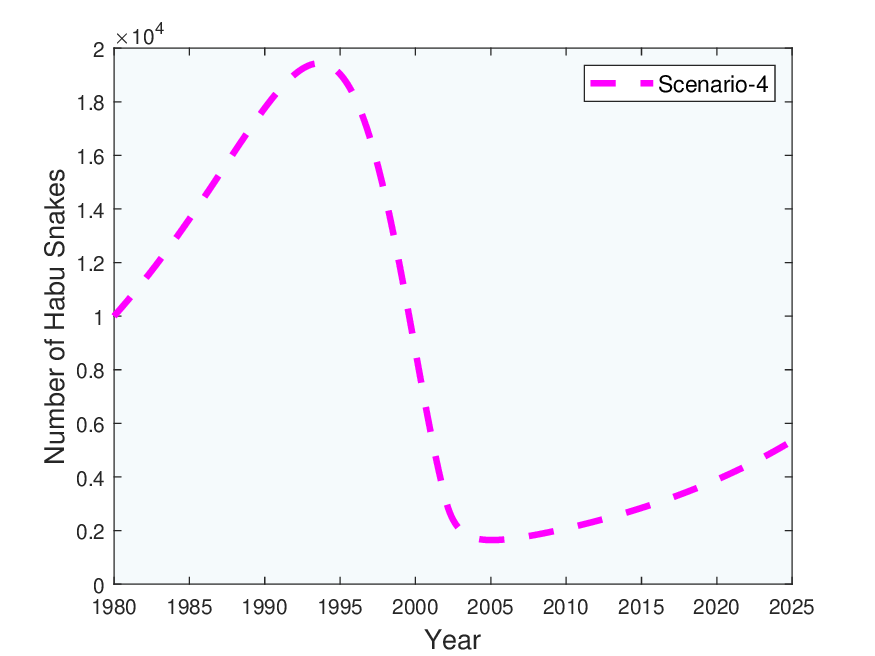}
        \caption{}
        \label{fig:fig33333}
    \end{subfigure}
    \hfill
    \begin{subfigure}{0.32\textwidth}
        \centering
        \includegraphics[width=\textwidth]{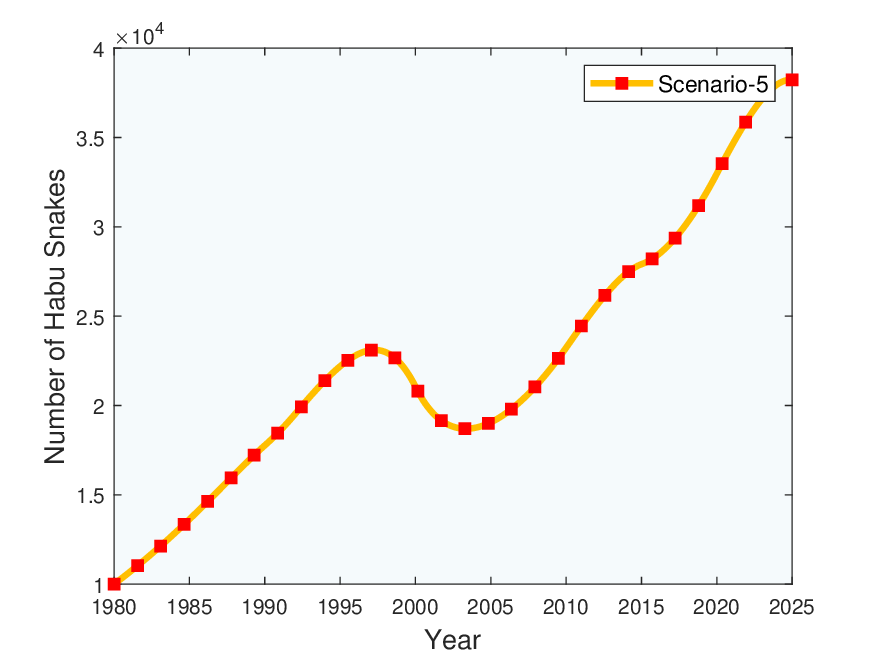}
        \caption{}
        \label{fig:fig444444}
    \end{subfigure}
    \hfill
    \begin{subfigure}{0.32\textwidth}
        \centering
        \includegraphics[width=\textwidth]{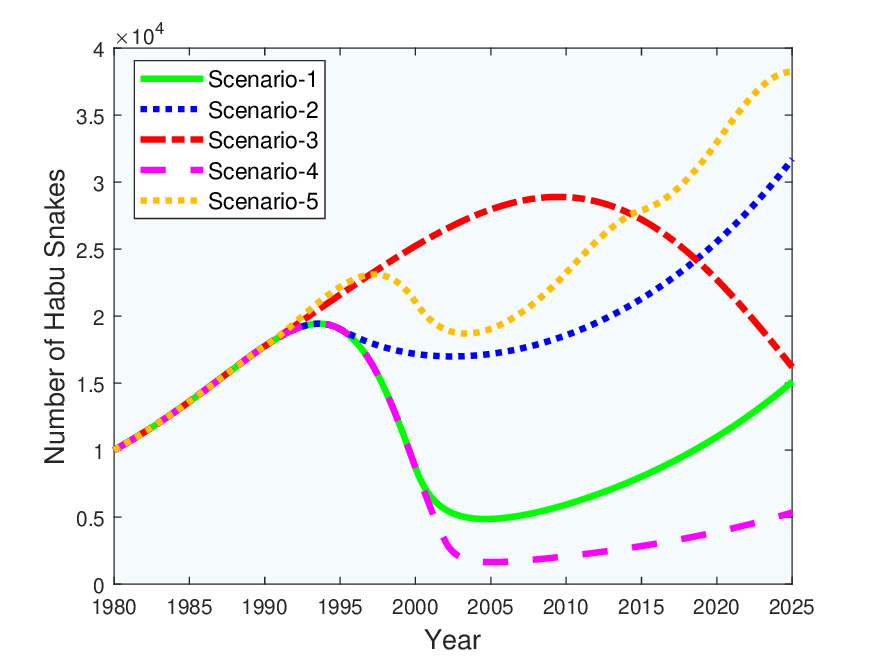}
        \caption{}
        \label{fig:fig5555555}
    \end{subfigure}
    \caption{Performed at Amami Oshima Island, the dynamics of the Habu snake population throughout time when (a) Mongoose trapping started in 2000 provide the current trend (Scenario-1), (b) When mongoose trapping began in 1995 at a slower pace than Scenario-1 (Scenario-2), (c) If mongoose trapping started in 1990 at a lesser pace than Scenario-1 \& 2 (Scenario-3), (d) Assuming mongoose trapping begins in 2002 at a rate faster than Scenario-1 (Scenario-4), (e) If mongoose trapping proceeded in three phases—1990–1992, 2000–2010, and 2015–2020 in Scenario-5 (f) merging all the scenarios in one frame.}
    \label{fig:Snake}
\end{figure}
\\
In this study, we evaluate five different mongoose control strategies, referred to as `Scenarios', each representing a distinct approach to trapping and removal. The trapping rate parameters, denoted as $\mu_1$ and $\mu_2$, vary across these scenarios to analyze their impact on ecosystem. The scenarios are designed as follows:  
\begin{itemize}
    \item \textbf{Scenario 1 ($\mu_1 = 0.5$, $\mu_2 = 0.2$):}  
  This scenario represents the actual mongoose eradication program implemented on Amami Ōshima Island. Trapping officially began in the year 2000, with an approximate daily removal of nine mongooses via netting and six through poisoning.

    \item \textbf{Scenario 2 ($\mu_1 = 0.25$, $\mu_2 = 0.1$):}  
    In this hypothetical scenario, mongoose trapping is assumed to have started earlier, in 1995, but at a lower removal rate compared to Scenario 1. The primary objective of this scenario is to investigate whether an earlier intervention, even with reduced trapping intensity, could have been more effective in preventing the severe ecological damage observed in later years.  

    \item \textbf{Scenario 3 ($\mu_1 = 0.15$, $\mu_2 = 0.09$):}  
    This scenario considers an even earlier intervention, beginning in 1990, with an even lower removal rate compared to Scenario 1 and Scenario 2. This scenario investigates whether an even earlier intervention, despite lower intensity, could have yielded better ecological outcomes.

    \item \textbf{Scenario 4 ($\mu_1 = 0.75$, $\mu_2 = 0.32$):}  
    Unlike the previous scenarios, this scenario assumes that mongoose trapping commenced in 2002, but at a significantly higher removal rate than in Scenario 1 and 2. This scenario examines the effect of a more aggressive eradication strategy.
    
    \item \textbf{Scenario 5 ($\mu_1 = 0.44$, $\mu_2 = 0.1$):}  
    This scenario introduces a phased trapping approach, where mongoose removal occurs in three distinct periods: 1990–1992, 2000–2010, and 2015–2020 with an approximate daily removal of five mongooses via netting and one through poisoning. Unlike a continuous eradication program, this approach assumes that removal efforts were strategically paused and resumed based on population trends.
\end{itemize}
The dynamics of mongooses, habu snakes, rats, natural resources, total population, and the number of captured mongooses over the time period for the considered Scenarios 1–5 are depicted in Figures \ref{fig:Mongooses}–\ref{fig:Total population} respectivly. Figure \ref{fig:Mongooses}(b) illustrates that under Scenario 3, the mongoose population began to decline gradually from 1995. However, by 2025, complete eradication was not seen, with approximately 500 mongooses still remaining. In contrast, an entirely opposite trend is observed in Scenario 3, as shown in Figure \ref{fig:Mongooses}(c). This suggests that if an even earlier intervention had been implemented in 1990 with a lower removal rate than in Scenarios 1 and 2, the mongoose population would have continued to rise rather than decline. Such an outcome would have posed a significant threat to the island’s ecosystem, exacerbating ecological imbalances. In Scenario 4, the mongoose population began to decline from 2002, as shown in Figure \ref{fig:Mongooses}(c), and was nearly eradicated by 2005. In contrast, Figure \ref{fig:Mongooses}(a) indicates that this process took longer in other scenarios.  Moreover, Figure \ref{fig:Mongooses}(e) reveals noticeable variations in the mongoose population at different time intervals, highlighting periods of growth and decline influenced by trapping efforts. From 1980, the mongoose population steadily increased until 1990. However, after two years, it declined due to trapping efforts. In the absence of sustained control measures, the population rebounded and continued to rise until 2000. A more intensive trapping program from 2000 to 2010 led to a significant decline, yet the population surged again once the trapping ceased. Another targeted trapping phase between 2015 and 2020 resulted in a further decline, but after this period, the mongoose population exhibited continuous growth. Overall, Figure \ref{fig:Mongooses}(f) presents a comparative view of mongoose population dynamics across all scenarios, showcasing how different intervention strategies impact their numbers over time. Furthermore, Figure \ref{fig:trapping} provides a detailed of the trapping efforts in each scenario in controlling the mongoose population on Amami Ōshima Island.\\
Figures \ref{fig:Snake} and \ref{fig:rat} illustrate the population dynamics of the habu snake and rats on Amami Ōshima Island across different scenarios, emphasizing their interactions with the mongoose population. Among them, in Figure \ref{fig:Snake}, the habu snake population declines as mongoose numbers rise due to predation but recovers when mongoose trapping reduces their numbers. The rate of recovery varies depending on the timing and intensity of intervention. Similarly, Figure \ref{fig:rat} shows fluctuations in the rat population, which increases when mongooses are removed but is later regulated by the recovering snake population.\\
Significant changes in natural resources presence over time are revealed by an in-depth analysis of the dynamics of natural resources under different scenarios, as shown in Figure \ref{fig:Natural Resources}. By the end of the period, Scenario 5 provides the maximum preservation of natural resources out of all the scenarios. This is consequently because it takes a more balanced approach to mongoose control, minimising undue ecological disturbance. The least amount of resources are available in Scenario 4, which may be the result of a more aggressive trapping technique that has unforeseen ecological repercussions. 
\\
As well, when trying to figure out the ecological balance of the island, it becomes important to understand the dynamics of the mongoose, snake, rat, and natural resource populations. This is because changes in one species, like over-trapping of mongooses, can cause a chain reaction effect on predator-prey dynamics and the availability of resources. These changes are shown in Figure \ref{fig:Total population}, which reflects the similar trends in Figure \ref{fig:Natural Resources} and underlines the effect of different strategies of intervention on the ecosystem.

Table \ref{tab:cumulative natural resources} now presents the cumulative natural resources obtained across five different mongoose trapping scenarios from 1980 to 2025, calculated using numerical integration with the Simpson $\frac{3}{8}$ rule. From this table, it is observed that Scenario-5, involving phased trapping, results in the highest resource retention (2,093,000), i.e almost 1.48\% more than the presence scenario (Scenario-1) in Amami Oshamia Island, whereas Scenario-4, with intensive trapping from 2002, leads to the lowest (1,167,800). These values are graphically represented in Figure \ref{fig:Impact of Mongoose Trapping on Natural Resources}, which clearly visualizes the occurance of different scenarios.

\begin{table}[hbt]
    \centering
     \caption{Mongoose trapping scenarios and the cumulative natural resources}
    \renewcommand{\arraystretch}{1.2}
    \begin{tabular}{p{1.5cm} p{6.3cm} p{2.8cm}}
        \hline
 \textcolor{blue}{ \textbf{Scenario No}} & \textcolor{blue}{\textbf{Description of Scenario}} & \textcolor{blue}{\textbf{Cumulative Natural Resources of}} \\
     & & \textcolor{blue}{\textbf{1980-2025}}\\
        \hline
        Scenario-1 & Mongoose trapping began in 2000, as performed in Amami Ōshima Island, which gives the present trend. & $1,414,700$ \\
        Scenario-2 & If mongoose trapping had started in 1995 at a lower rate than Scenario-1. & $1,920,400$ \\
        Scenario-3 & If mongoose trapping were initiated in 1990 at a lower rate than both Scenario-1 and Scenario-2. & $2,049,200$ \\
        Scenario-4 & If mongoose trapping could have been launched in 2002 at a higher rate than Scenario-1. & $1,167,800$ \\
        Scenario-5 & If mongoose trapping could have begun in three phases: 1990–1992, 2000–2010, and 2015–2020. & $2,093,000$ \\
        \hline
    \end{tabular}
    \label{tab:cumulative natural resources}
\end{table}

\begin{figure}[!ht]
    \centering
 \begin{subfigure}{0.32\textwidth}
        \centering
        \includegraphics[width=\textwidth]{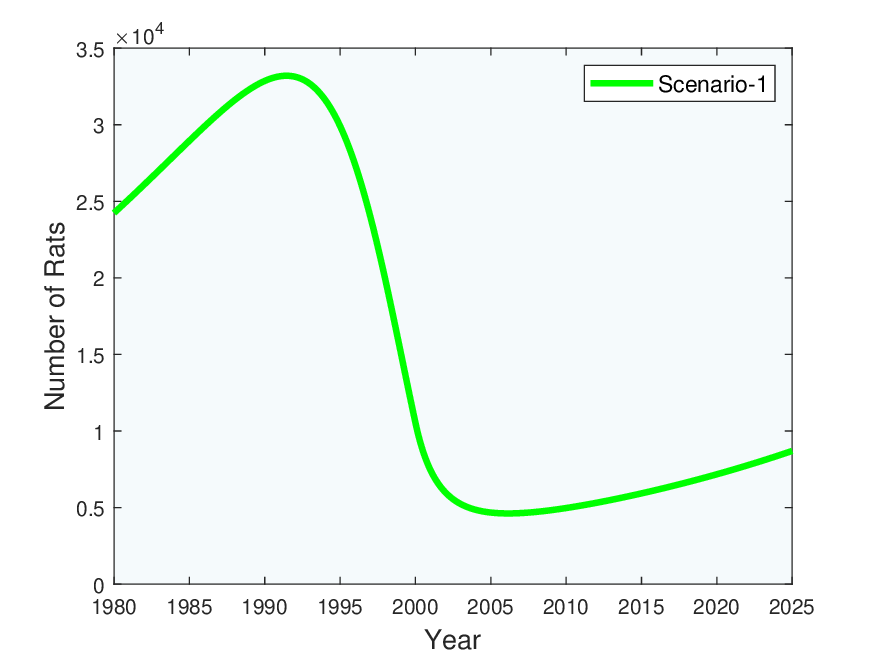}
        \caption{}
    \end{subfigure}
    \hfill
    \begin{subfigure}{0.32\textwidth}
        \centering
        \includegraphics[width=\textwidth]{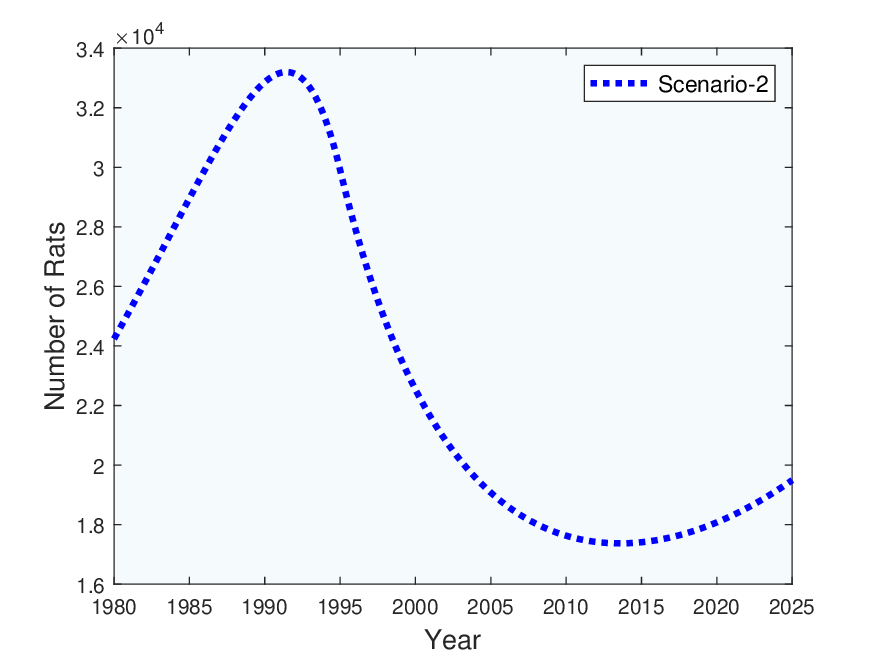}
        \caption{}
    \end{subfigure}
    \hfill
    \begin{subfigure}{0.32\textwidth}
        \centering
        \includegraphics[width=\textwidth]{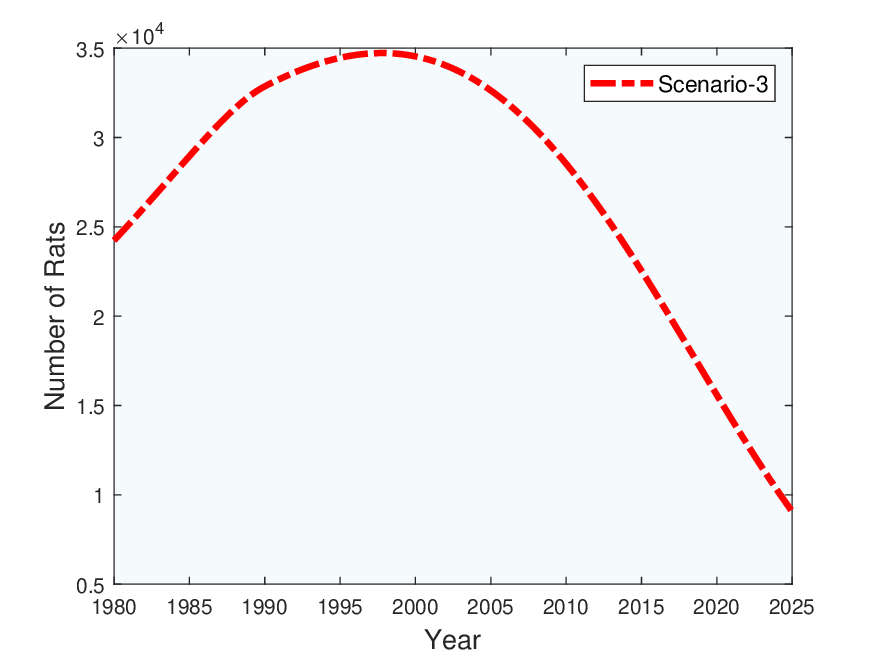}
        \caption{}
    \end{subfigure}
    \hfill
    \begin{subfigure}{0.33\textwidth}
        \centering
        \includegraphics[width=\textwidth]{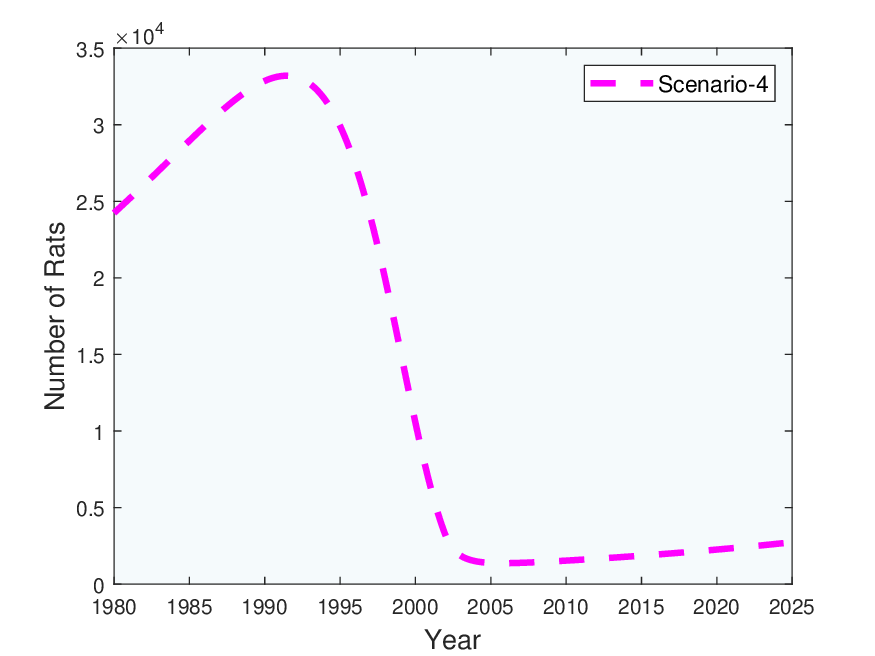}
        \caption{}
    \end{subfigure}
    \hfill
    \begin{subfigure}{0.32\textwidth}
        \centering
        \includegraphics[width=\textwidth]{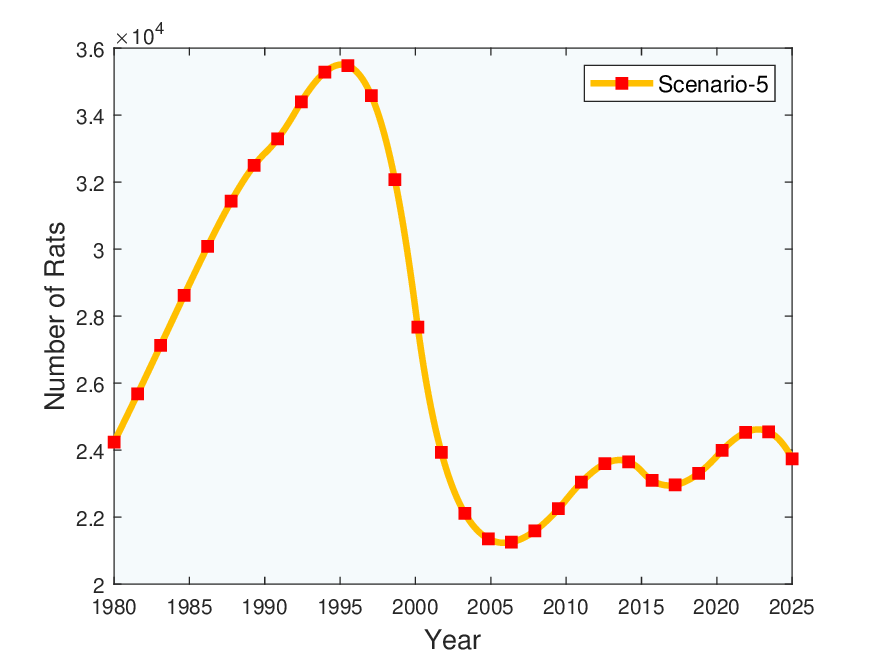}
        \caption{}
    \end{subfigure}
    \hfill
    \begin{subfigure}{0.32\textwidth}
        \centering
        \includegraphics[width=\textwidth]{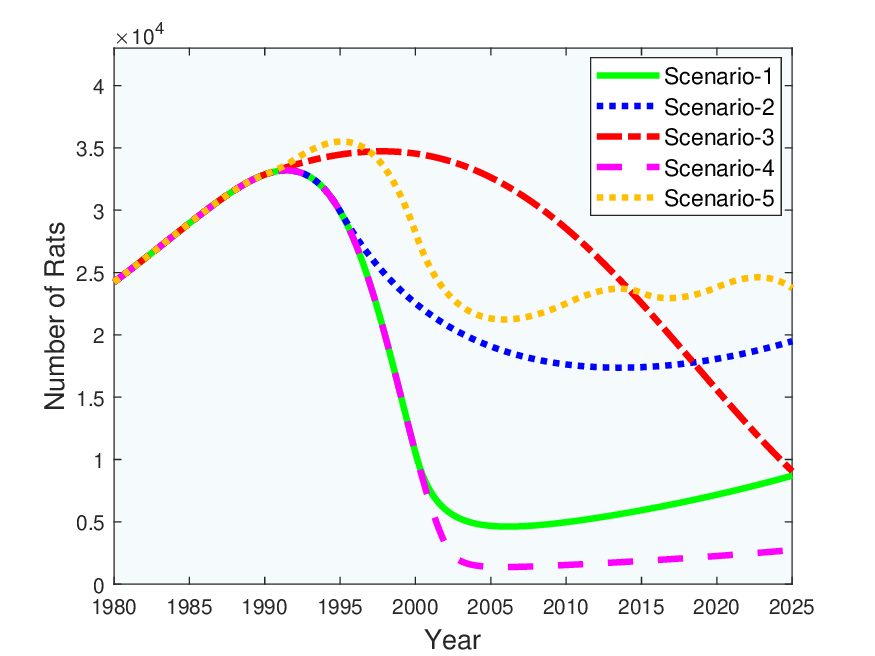}
        \caption{}
    \end{subfigure}
    \caption{Rat Population Dynamics in Amami Ōshima Under Various Mongoose Trapping Strategies (a) Scenario-1: Actual trend observed when mongoose trapping began in 2000.
(b) Scenario-2: Outcome if trapping had started in 1995 at a lower intensity than in Scenario-1 (c) Scenario-3: Projected impact if trapping had commenced in 1990 at an even lower rate than in Scenario-1 and Scenario-2 (d) Scenario-4: Potential changes if trapping had been introduced in 2002 with a higher intensity than in Scenario-1
(e) Scenario-5: A phased trapping approach implemented over three periods—1990–1992, 2000–2010, and 2015–2020 (f) A comparative visualization combining all scenarios in a single frame}
    \label{fig:rat}
\end{figure}

\begin{figure}[hbtp]
    \centering
 \begin{subfigure}{0.32\textwidth}
        \centering
        \includegraphics[width=\textwidth]{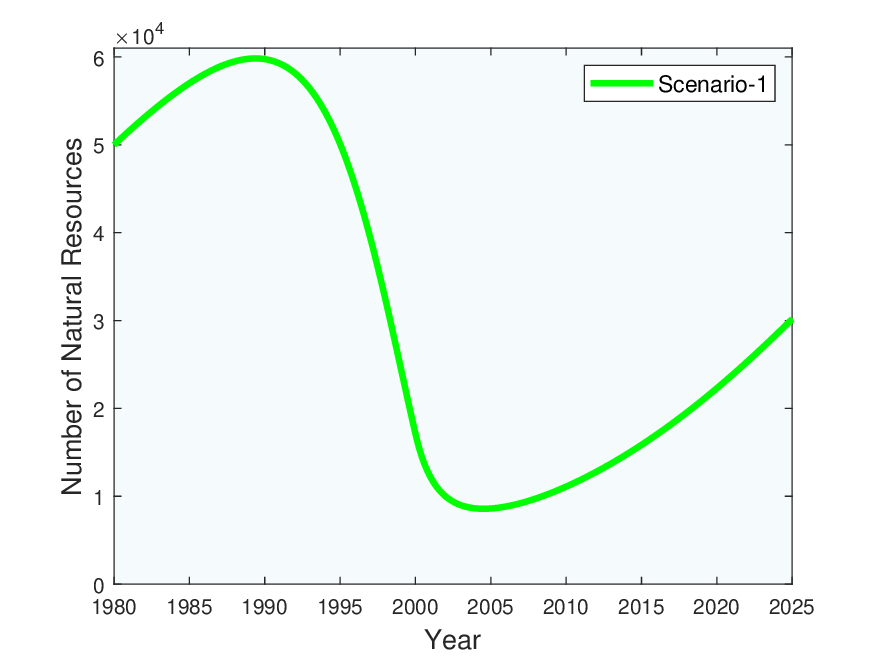}
        \caption{}
    \end{subfigure}
    \hfill
    \begin{subfigure}{0.32\textwidth}
        \centering
        \includegraphics[width=\textwidth]{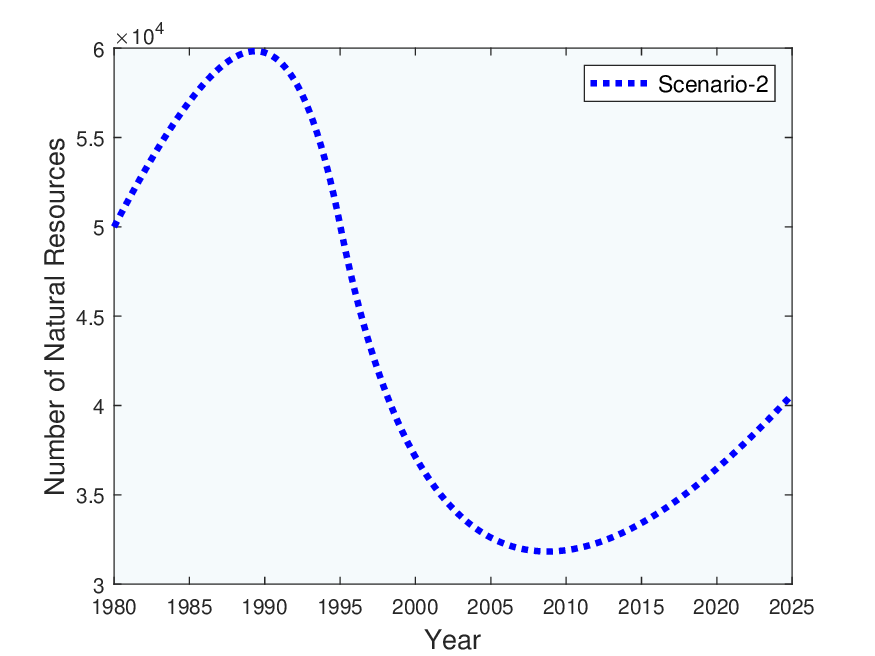}
        \caption{}
    \end{subfigure}
    \hfill
    \begin{subfigure}{0.32\textwidth}
        \centering
        \includegraphics[width=\textwidth]{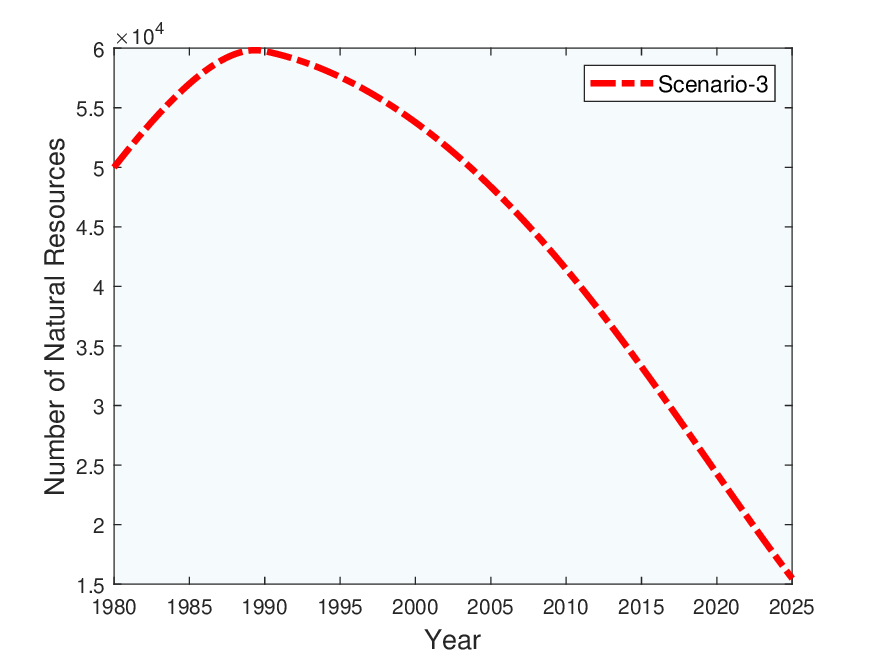}
        \caption{}
    \end{subfigure}
    \hfill
    \begin{subfigure}{0.33\textwidth}
        \centering
        \includegraphics[width=\textwidth]{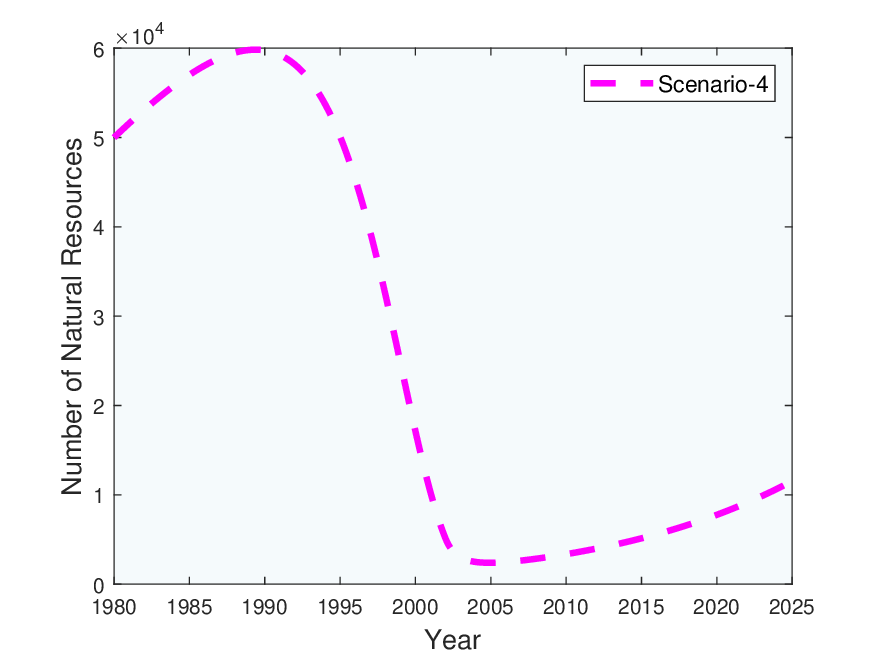}
        \caption{}
    \end{subfigure}
    \hfill
    \begin{subfigure}{0.32\textwidth}
        \centering
        \includegraphics[width=\textwidth]{figures/S5_R.eps}
        \caption{}
    \end{subfigure}
    \hfill
    \begin{subfigure}{0.32\textwidth}
        \centering
        \includegraphics[width=\textwidth]{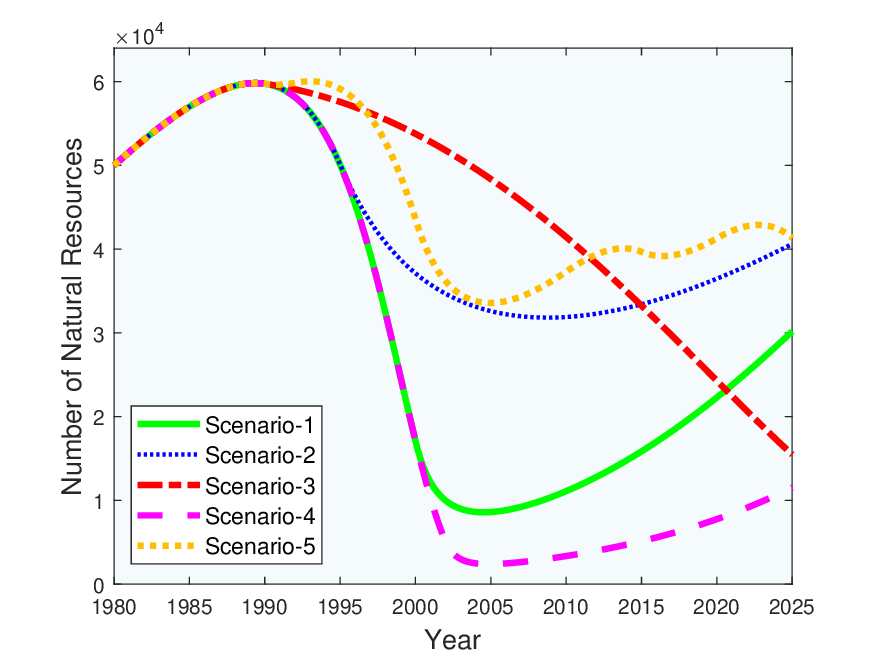}
        \caption{}
    \end{subfigure}
    \caption{The changing dynamics of natural resources in Amami Oshima Island under various mongoose trapping scenarios: (a) The actual trapping effort, initiated in 2000, representing the current trend (Scenario 1) (b) A variational change where trapping began in 1995 at a lower intensity than previous (Scenario 2) (c) An earlier intervention starting in 1990 with a lower trapping rate than both previous two (Scenario 3)  (d) A delayed but more aggressive trapping effort launched in 2002 at a higher rate than current trend (Scenario 4)  (e) A phased trapping strategy implemented in three stages: 1990–1992, 2000–2010, and 2015–2020 (Scenario 5) . (f) A combined visualization of all scenarios for comparative analysis. }
    \label{fig:Natural Resources}
\end{figure}
\begin{figure}[hbtp]
    \centering
 \begin{subfigure}{0.32\textwidth}
        \centering
        \includegraphics[width=\textwidth]{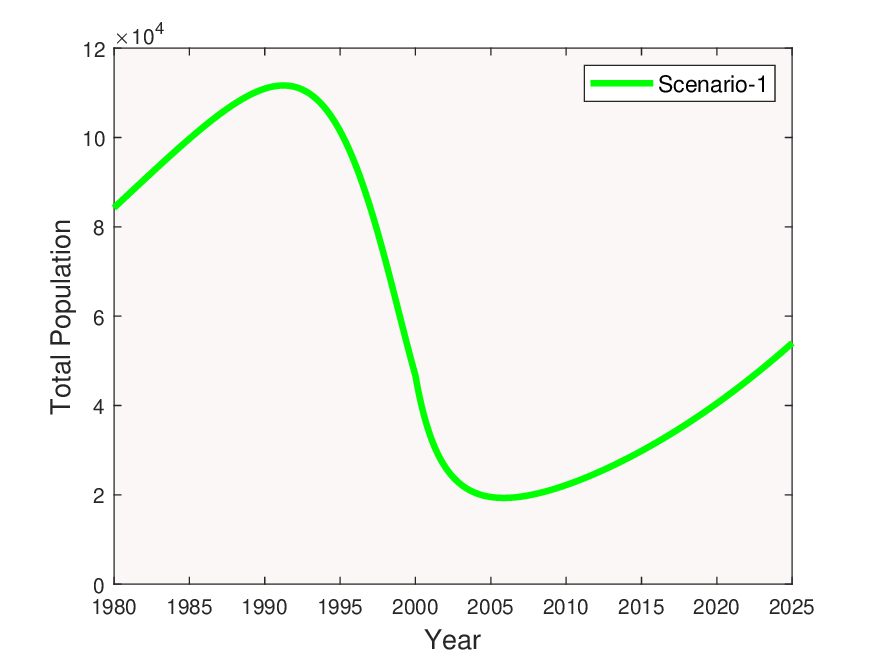}
        \caption{}
    \end{subfigure}
    \hfill
    \begin{subfigure}{0.32\textwidth}
        \centering
        \includegraphics[width=\textwidth]{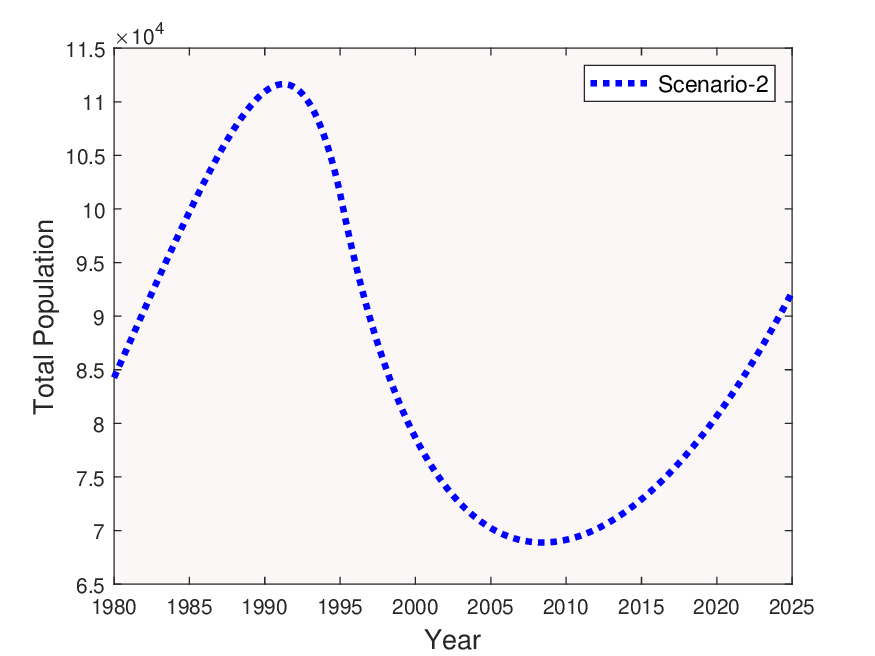}
        \caption{}
    \end{subfigure}
    \hfill
    \begin{subfigure}{0.32\textwidth}
        \centering
        \includegraphics[width=\textwidth]{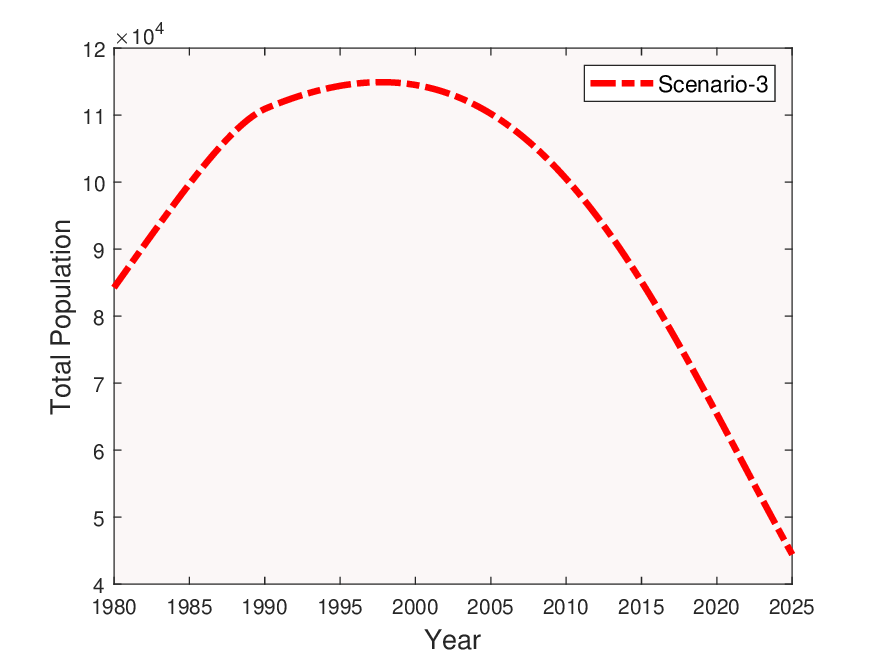}
        \caption{}
    \end{subfigure}
    \hfill
    \begin{subfigure}{0.33\textwidth}
        \centering
        \includegraphics[width=\textwidth]{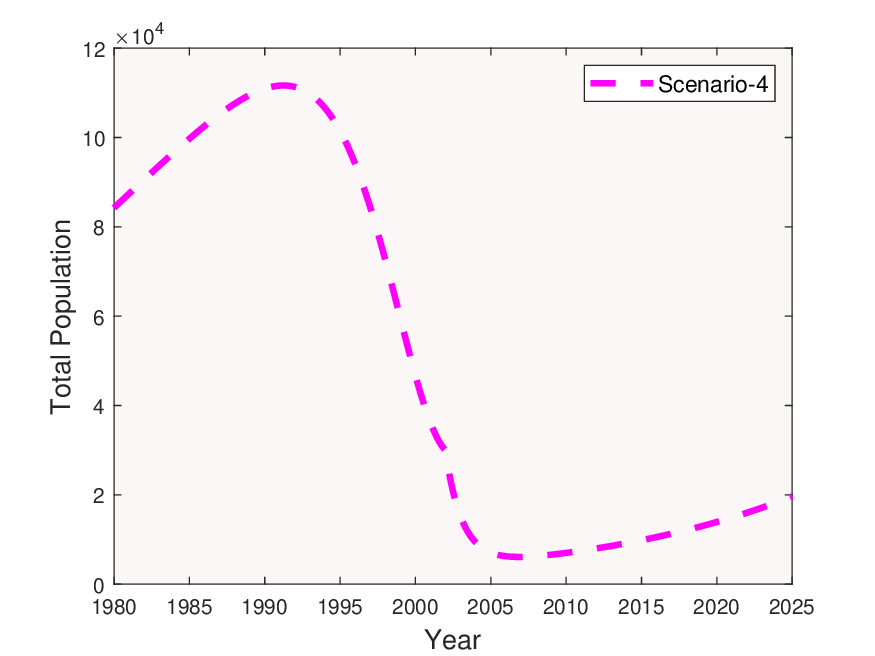}
        \caption{}
    \end{subfigure}
    \hfill
    \begin{subfigure}{0.32\textwidth}
        \centering
        \includegraphics[width=\textwidth]{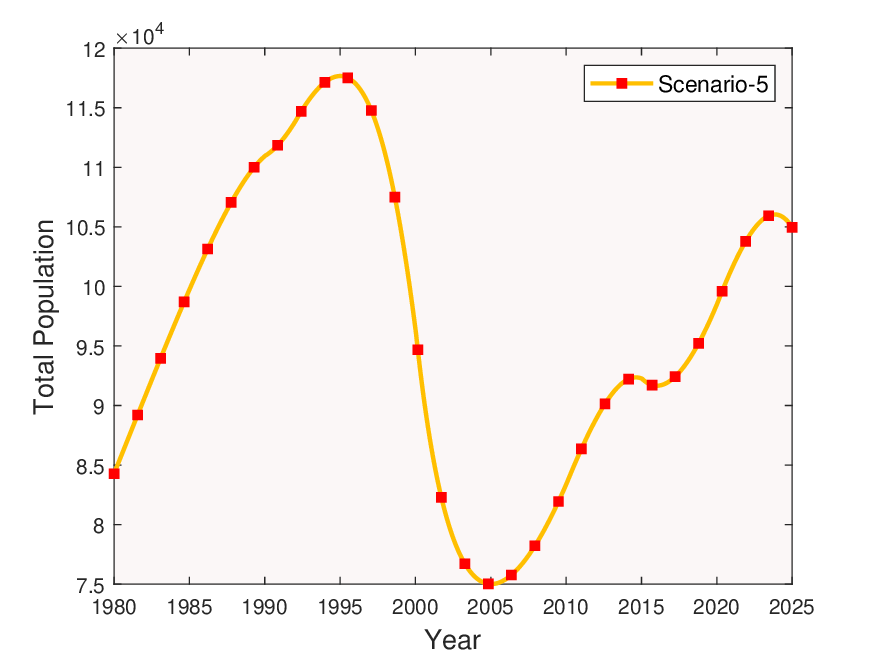}
        \caption{}
    \end{subfigure}
    \hfill
    \begin{subfigure}{0.32\textwidth}
        \centering
        \includegraphics[width=\textwidth]{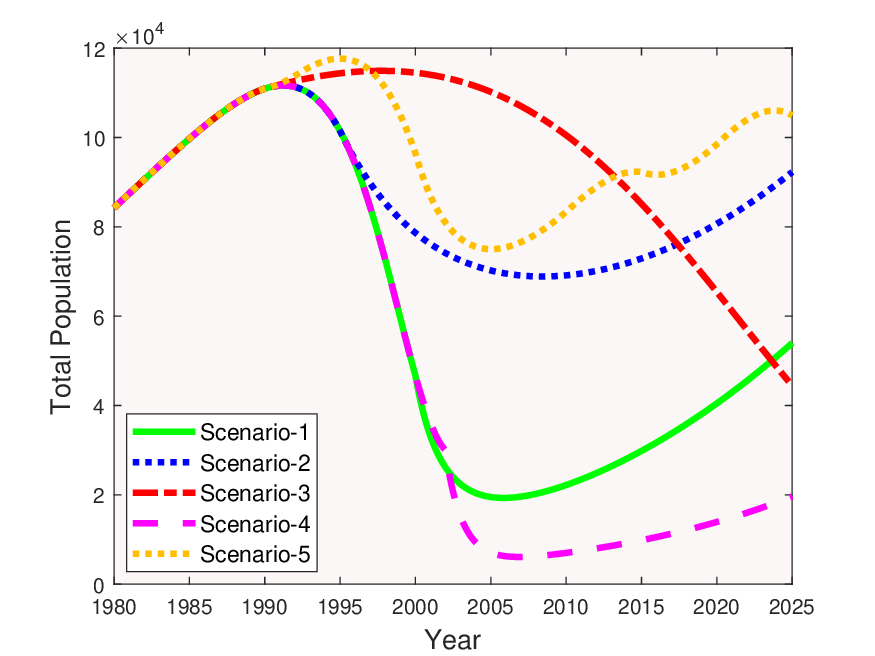}
        \caption{}
    \end{subfigure}
    \caption{The total population dynamics of mongooses, snakes, rats, and natural resources on Amami Oshima Island under different mongoose trapping scenarios: (a) The actual trapping effort initiated in 2000, reflecting the current trend (Scenario 1) (b) An alternative approach with trapping beginning in 1995 at a lower intensity than the current trend (Scenario 2) (c) An earlier intervention starting in 1990 with a trapping rate lower than both previous scenarios (Scenario 3) (d) A delayed yet more aggressive trapping effort launched in 2002 at a higher rate than the current trend (Scenario 4) (e) A phased trapping strategy implemented in three stages: 1990–1992, 2000–2010, and 2015–2020 (Scenario 5) (f) A comprehensive comparison displaying all scenarios together (Scenario 1-5)}
    \label{fig:Total population}
\end{figure}
\begin{figure}[hbtp]
    \centering
 \begin{subfigure}{0.32\textwidth}
        \centering
        \includegraphics[width=\textwidth]{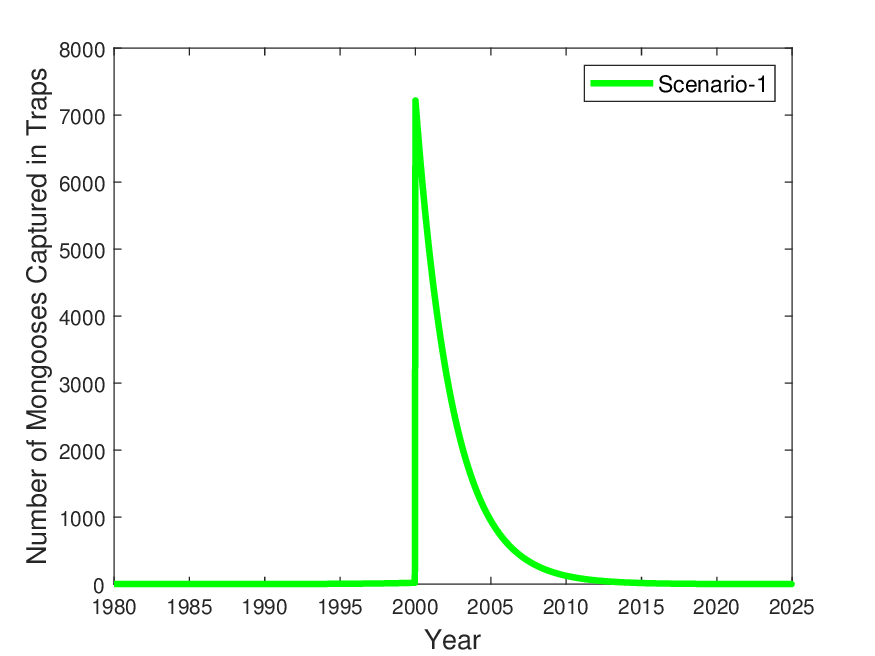}
        \caption{}
    \end{subfigure}
    \hfill
    \begin{subfigure}{0.32\textwidth}
        \centering
        \includegraphics[width=\textwidth]{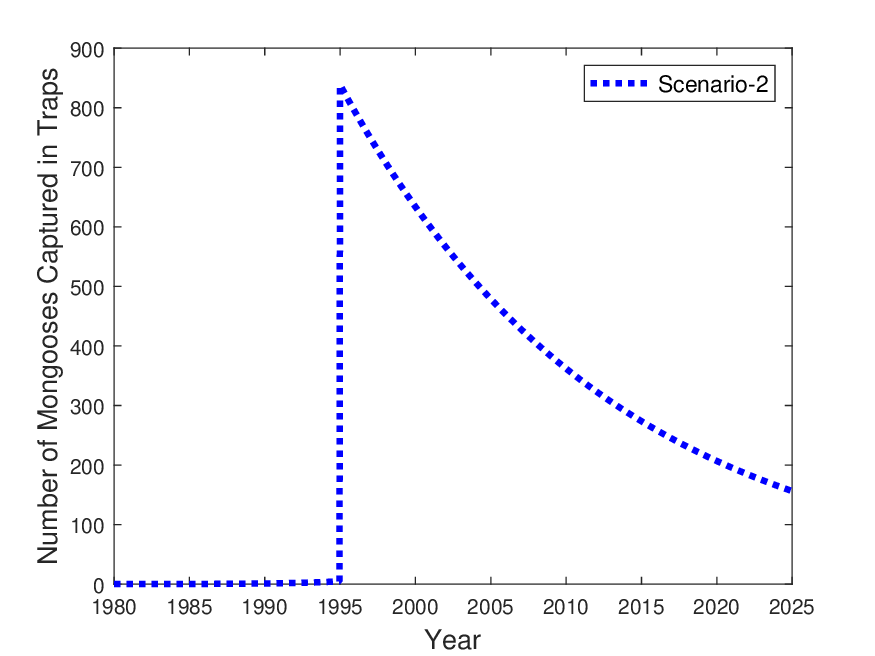}
        \caption{}
    \end{subfigure}
    \hfill
    \begin{subfigure}{0.32\textwidth}
        \centering
        \includegraphics[width=\textwidth]{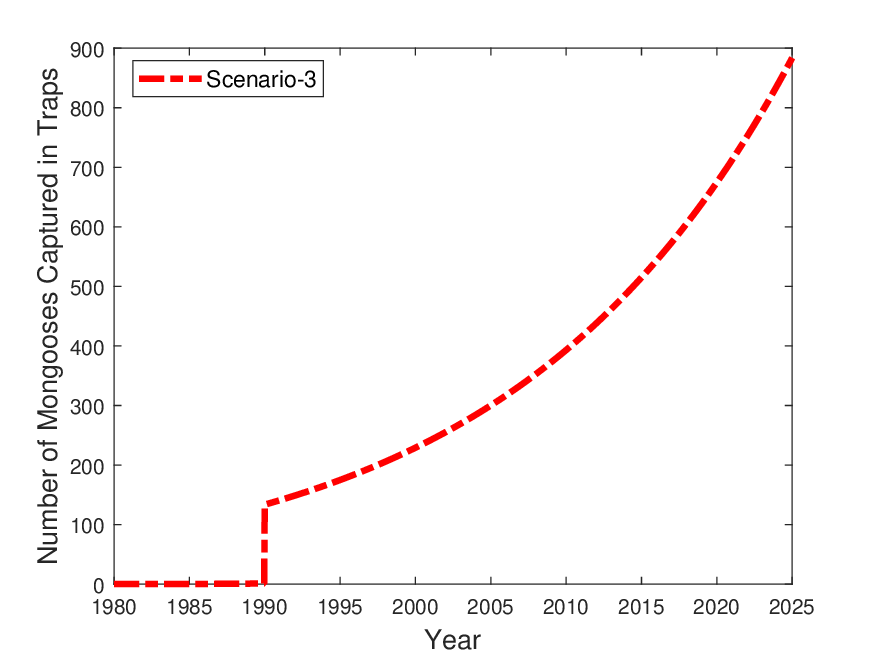}
        \caption{}
    \end{subfigure}
    \hfill
    \begin{subfigure}{0.33\textwidth}
        \centering
        \includegraphics[width=\textwidth]{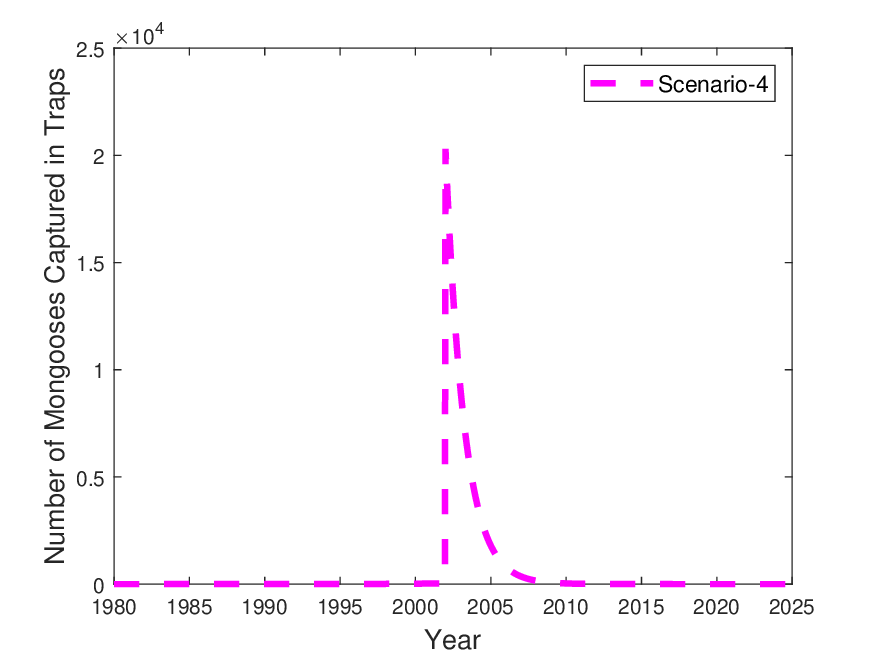}
        \caption{}
    \end{subfigure}
    \hfill
    \begin{subfigure}{0.32\textwidth}
        \centering
        \includegraphics[width=\textwidth]{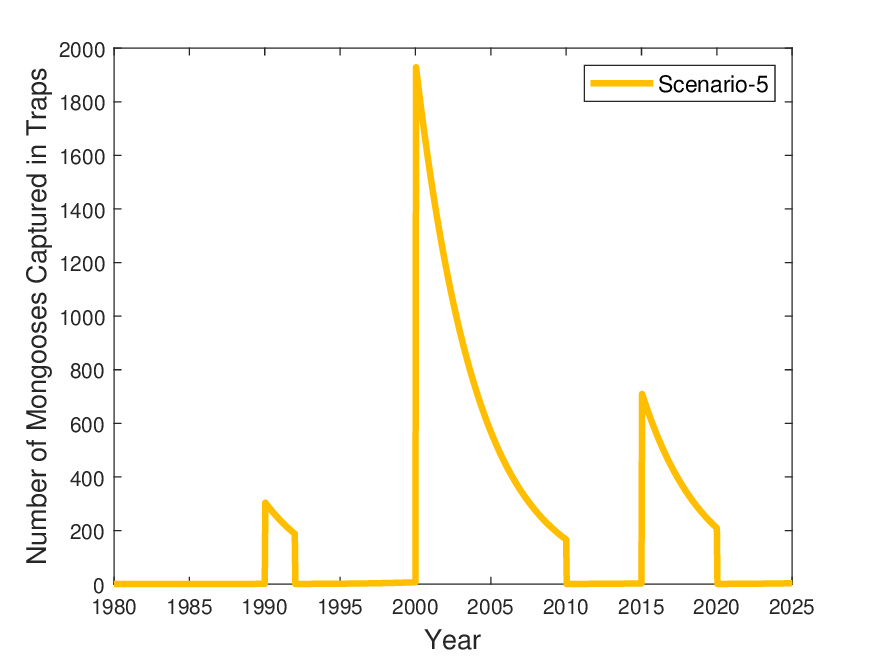}
        \caption{}
    \end{subfigure}
    \hfill
    \begin{subfigure}{0.32\textwidth}
        \centering
        \includegraphics[width=\textwidth]{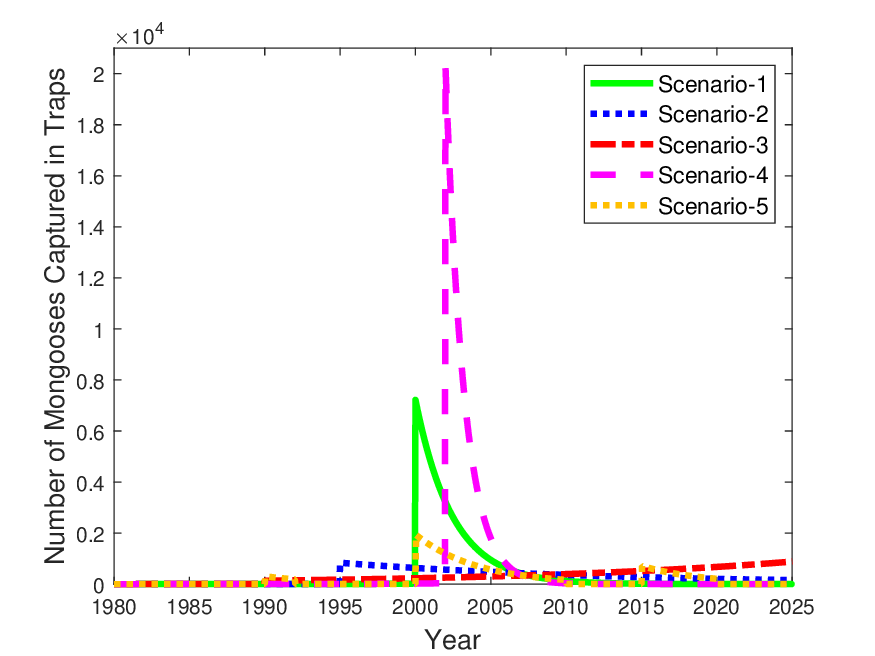}
        \caption{}
    \end{subfigure}
    \caption{Patterns of mongoose captures in traps on Amami Ōshima Island under different trapping scenarios: (a) The actual trapping effort, initiated in 2000, reflects the current trend (Scenario 1) (b) A modified approach where trapping began in 1995 at a lower intensity than the current strategy (Scenario 2) (c) An earlier intervention starting in 1990 with a lower trapping rate than both previous scenarios (Scenario 3) (d) A delayed yet more intensive trapping effort, launched in 2002 at a higher rate than the current trend (Scenario 4) (e) A phased trapping strategy implemented in three distinct periods: 1990–1992, 2000–2010, and 2015–2020 (Scenario 5) (f) A comprehensive visualization comparing all scenarios to assess their relative impact}
    \label{fig:trapping}
\end{figure}

\begin{figure}[hbtp]
    \centering
    \includegraphics[width=0.8\textwidth, keepaspectratio]{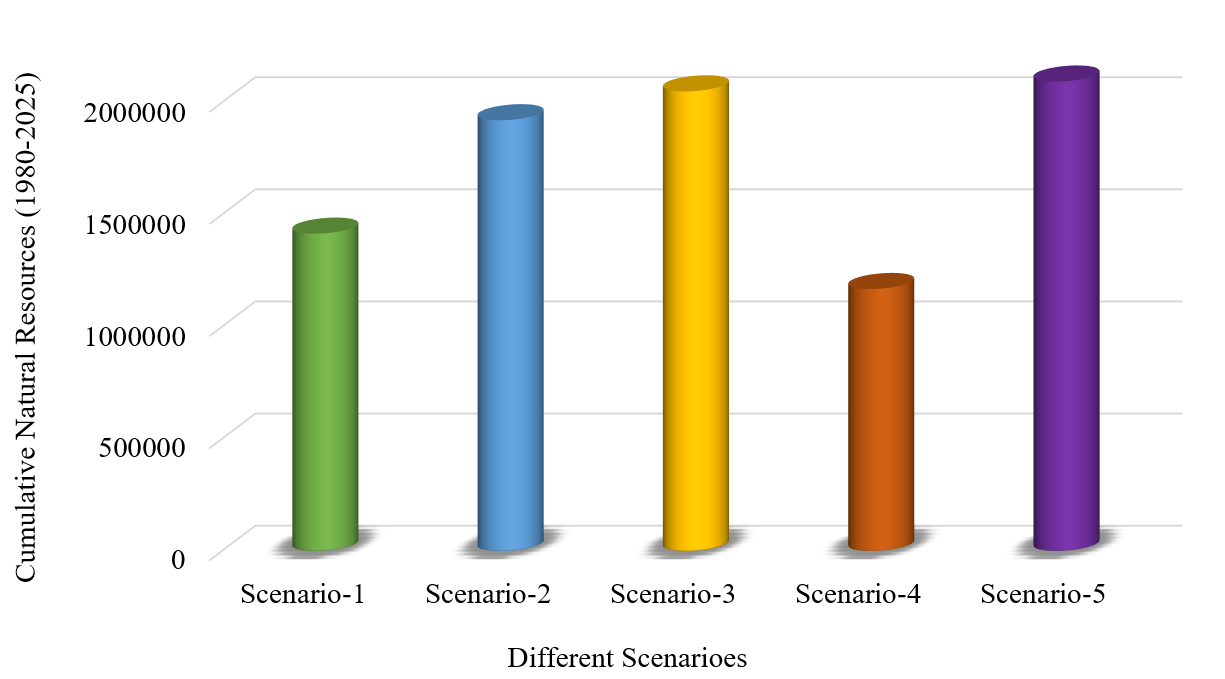} 
    \caption{\centering Impact of Mongoose Trapping on Natural Resources.}
    \label{fig:Impact of Mongoose Trapping on Natural Resources}
\end{figure}

\section{Results and Discussion}\label{sec:Results and disussion}
The introduction of mongooses, which were initially intended for controlling the habu snake population, has adversely destroyed the biological balance of Amami Ōshima Island, Japan. The island's ecology is now unsustainable and natural resources are in danger as a result of this intervention's undesired consequences. Understanding the long-term prey-predator dynamics is important for devising effective conservation strategies for a ecosystem.\\
Figure \ref{fig:solution} presents the real scenario of mongooses, habu snakes, rats, and natural resources on Amami Ōshima Island, obtained from the numerical solution of our proposed model \eqref{eq:model}. It highlights significant declines in natural resources, posing a threat to the ecological balance for the today time. Besides, Figure \ref{fig:alpha_variation} reflects that a higher growth rate of mongooses would be harmful to the island's ecological balance as it leads to an increased predation pressure on native species and a further depletion of natural resources. Also, the relationship between different variables illustrates in Figure \ref{fig:phase diagram} suggested the changing pattern of the different species over the time interval with one another. Besides, Figures \ref{fig:Mongooses}-\ref{fig:Impact of Mongoose Trapping on Natural Resources} show how these populations behave in various situations. These figures demonstrate that trapping would have successfully decreased the mongoose population without precipitously depleting natural resources if it had been carried out in three periods (1990–1992, 2000–2010, and 2015–2020) at a lower rate(about six mongooses per day). By allowing the native species to persist and explode, this progressive management would have preserved the island's ecological balance and shown how effective mathematical modelling is in predicting and analysing the long-term impacts of predator-prey interactions.

\section{Conclusions}\label{sec:Conclusions}
This study introduces a newly proposed mathematical model to examine the prey-predator dynamics and their impact on natural resources on Amami Ōshima Island, emphasizing the significant ecological consequences of introducing mongooses to control the life-threatening habu snake population. Theoretical and numerical analyses were performed to explore these interactions in detail. As part of the theoretical evaluation, we established that the model is bounded, ensuring its feasibility, and analyzed the stability of equilibrium points. Furthermore, different equilibrium values for a parameter have been determined, yielding important findings into the impact of many influences on the island's balance of ecology.\\
Based on the results of the numerical simulations, it appears that a controlled rate of phased trapping might successfully regulate species populations and conserve natural resources, hence restoring ecological balance. These findings highlight the crucial role of mathematical modeling in understanding complex ecological interactions and formulating well-informed, effective conservation strategies. Moreover, the study underscores the broader applicability of these methods to other regions facing similar ecological challenges, offering valuable insights for sustainable ecosystem management.\\
Future research can focus on exploring the cost of optimizing natural resource levels through optimal control, reflecting on the efforts that should have been required to be implemented in the past to sustain ecological balance on Amami Ōshima Island, Japan.

\section*{Author Contribution}
\textbf{Pulak Kundu}: Conceptualization, Formal analysis, Methodology, Data collection, Software, Writing –original draft, Writing –review \& editing. \\ \textbf{Uzzwal Kumar Mallick}: Conceptualization, Formal analysis, Methodology, Data collection, Software,  Writing –review \& editing, Supervision, Validation.
\section*{Declarations}
\textbf{Ethical Approval}\qquad Not applicable \\ \\
\textbf{Availability of Data and Material}\qquad The data used to support the findings of this study are included in the manuscript. \\ \\
\textbf{Conflicts of Interest}\qquad The authors declare that they have no known competing financial interests or personal relationships that could have appeared
to influence the work reported in this paper. No specific grant for this research was provided by funding organizations in the public,
private, or not-for-profit sectors \\ \\
\textbf{Acknowledgement}\qquad We would like to express our gratitude towards reviewers and editors for their insightful suggestions that have enhanced the quality and authenticity of our work.






\bibliography{sn-bibliography}

\end{document}